\documentclass[oribibl]{llncs}
\usepackage{llncsdoc}

\usepackage{graphicx}
\usepackage{verbatim}
\usepackage{amssymb}
\usepackage{amsmath}
\usepackage{subfig}
\usepackage{caption}
\usepackage[ruled,vlined]{algorithm2e}

\newcommand{\mA}{{\cal A}}

\newcommand{\mC}{{\cal C}}

\newcommand{\mK}{{\cal K}}
\newcommand{\mL}{{\cal L}}

\newcommand{\mP}{{\cal P}}

\newcommand{\mR}{{\cal R}}

\newcommand{\mT}{{\cal T}}

\newcommand{\nosemic}{\renewcommand{\@endalgocfline}{\relax}}
\newcommand{\dosemic}{\renewcommand{\@endalgocfline}{\algocf@endline}}

\begin{document}

\title{\Large Linear-Time Algorithms for Finding Tucker Submatrices and Lekkerkerker-Boland Subgraphs}

\institute{Colorado State University, Fort Collins CO 80521, USA}

\author{Nathan Lindzey~\thanks{lindzey@cs.colostate.edu,
Computer Science Department,
Colorado State University,
Fort Collins, CO, 80523-1873
U.S.A.},
Ross M. McConnell~\thanks{rmm@cs.colostate.edu,
Computer Science Department,
Colorado State University,
Fort Collins, CO, 80523-1873
U.S.A.}
}

\date{}

\maketitle

\begin{abstract}
Lekkerkerker and Boland characterized the minimal forbidden induced subgraphs 
for the class of interval graphs.  We give a linear-time algorithm to
find one in any graph that is not an interval graph.  Tucker characterized the
minimal forbidden submatrices of binary matrices that do not have the consecutive-ones
property.  We give a linear-time algorithm to find one in any binary matrix that
does not have the consecutive-ones property.
\end{abstract}

\section{Introduction}

The {\em intersection graph} of a collection of sets has one vertex for each
set in the collection and an edge between two vertices if the corresponding sets
intersect.
A graph is an {\em interval graph} if it is the intersection graph of a collection of intervals
on a line.
Such a collection of intervals is known as an {\em interval model} of the graph.
Interval graphs are an important subclass of perfect graphs~\cite{Gol80}, they have been written
about extensively, and they model constraints in various combinatorial optimization
and decision problems~\cite{CLRS09, Gol80, Roberts78, SpinGR}.   They have a rich structure and history,
and interesting relationships to other graph classes.  For a survey, see~\cite{BLS99}.

For a 0-1 (binary) matrix $M$, let $n$ denote the number of rows, $m$ the number
of columns, and $size(M)$ the number
of rows, columns and 1's.  In this paper, all matrices are binary.
A sparse representation of a matrix takes $O(size(M))$ space.
Such a matrix has the {\em consecutive-ones property} if there exists an ordering
of its columns such that, in every row, the 1's are consecutive.

A {\em consecutive-ones matrix} is a matrix that has the consecutive-ones property,
and a {\em consecutive-ones-ordered matrix} is a matrix where the 1's are consecutive
in every row.
A {\em clique} is a maximal induced subgraph that is a complete graph.
A {\em clique matrix} of a graph $G$ is a matrix that has a row for each
vertex, a column for each clique, and a 1 in row $i$, column $j$ if vertex $i$ is contained
in clique $j$.  
A graph is an interval graph if and only if its clique matrices have the consecutive-ones
property, see, for example,~\cite{Gol80}.

In 1962, Lekkerkerker and Boland described the minimal induced forbidden subgraphs for
the class of interval graphs~\cite{LB62}, known as the {\em LB graphs}.
These are depicted in Figure~\ref{fig:LBSubgraphs}.
Ten years later,
Tucker described the minimal forbidden submatrices for consecutive-ones matrices~\cite{Tuck72}.
These are the matrices that do not have the consecutive-ones property,
but where deletion of any row or column results in a matrix that has the 
consecutive-ones property.
These are depicted in Figure~\ref{fig:TuckMatrices}.  The presence of an induced LB subgraph
in a graph and the presence of a Tucker submatrix in its clique matrix are both necessary
and sufficient conditions for a graph not to be an interval graph.  Not surprisingly, there is a relationship
between the Tucker matrices and the clique matrices of LB graphs, depicted in Figure~\ref{fig:LBGraphsTucker}.

\begin{figure}[!ht]
\captionsetup[subfigure]{labelformat=empty}
\centering
\subfloat[$G_I$]{\includegraphics[scale=.25]{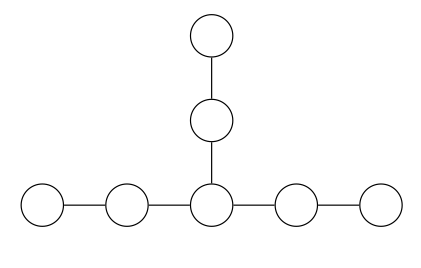}}
\subfloat[$G_{II}$]{\includegraphics[scale=.25]{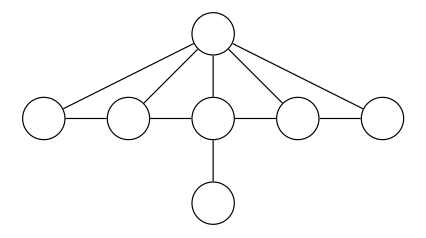}}
\subfloat[$G_{III}(n)$, $n \geq 4$]{\includegraphics[scale=.25]{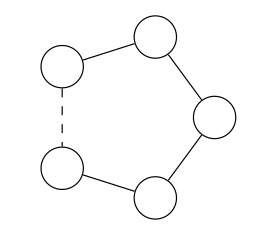}}\\
\subfloat[$G_{IV}(n)$, $n \geq 6$]{\includegraphics[scale=.25]{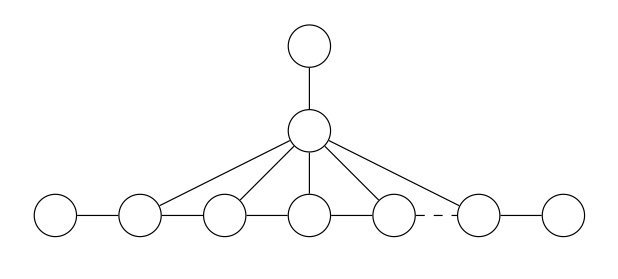}}
\subfloat[$G_V(n)$, $n \geq 6$]{\includegraphics[scale=.25]{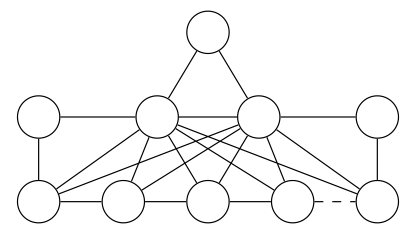}}

\caption{The Lekkerkerker-Boland subgraphs.}\label{fig:LBSubgraphs}
\end{figure}

\begin{figure}[!ht]
\captionsetup[subfigure]{labelformat=empty}
\centering
\subfloat[]{\includegraphics[scale=.25]{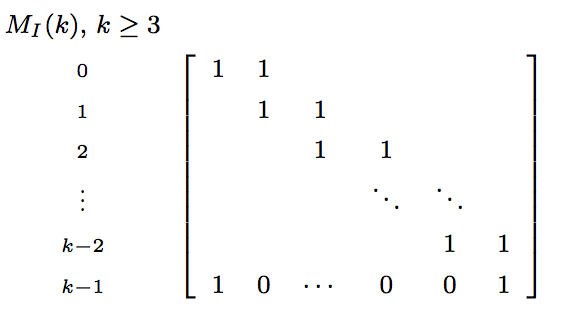}}
\quad \quad \quad
\subfloat[]{\includegraphics[scale=.22]{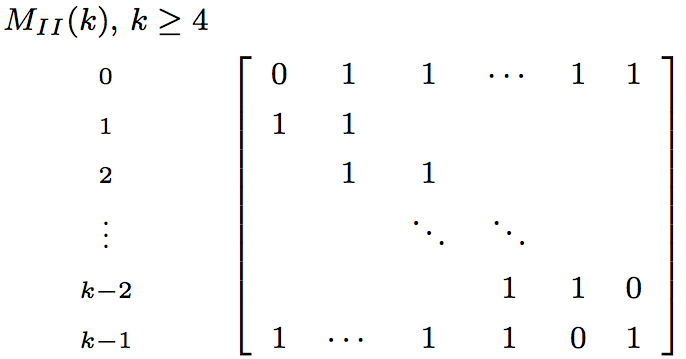}}\\
\centering
\subfloat[]{\includegraphics[scale=.25]{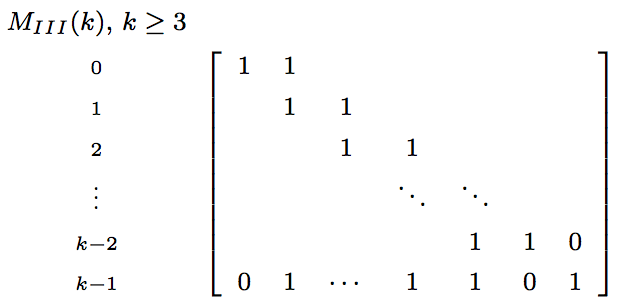}}\\
\quad \quad
\subfloat[]{\includegraphics[scale=.35]{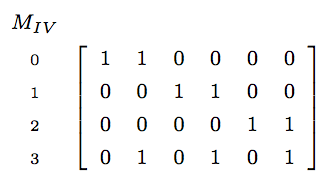}}
\quad \quad
\subfloat[]{\includegraphics[scale=.23]{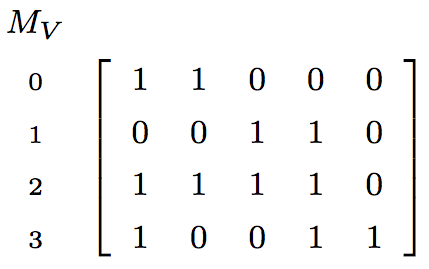}}

\caption{The minimal forbidden submatrices for consecutive-ones matrices.  
Entries that have nothing in them are implicitly 0's.
Since they are minimal matrices that cannot be consecutive-ones ordered,
for any chosen ordering of the rows, all rows except the last can be
consecutive-ones ordered.
}\label{fig:TuckMatrices}
\end{figure}

\begin{figure}[!ht]
\captionsetup[subfigure]{labelformat=empty}
\centering
\subfloat[$G_I$]{\includegraphics[scale=.18]{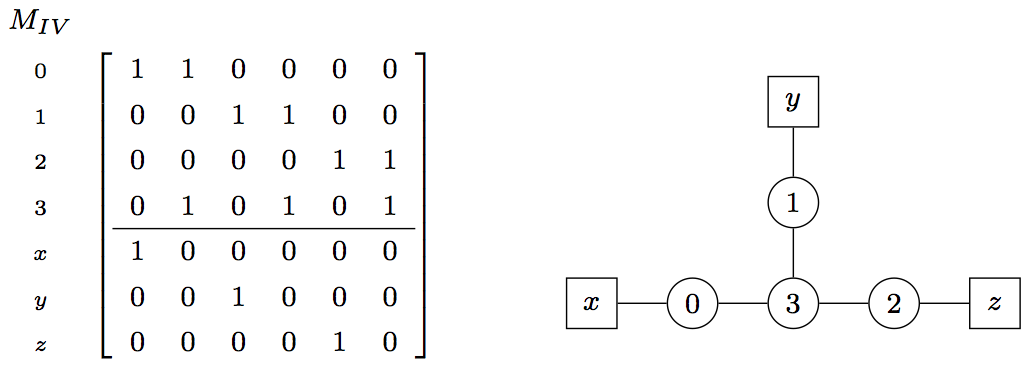}}
\subfloat[$G_{II}$]{\includegraphics[scale=.18]{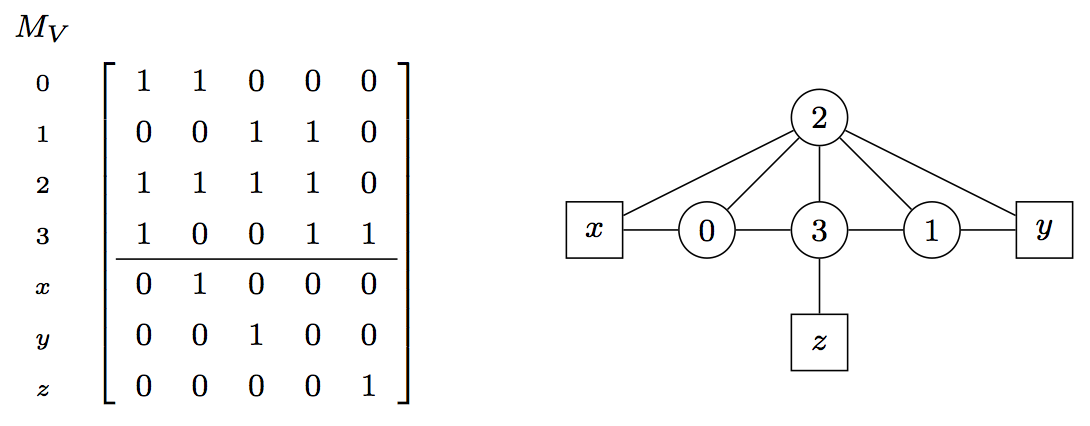}}\\

\subfloat[$G_{III}(n)$, $n \geq 4$, $k \geq 4$]{\includegraphics[scale=.21]{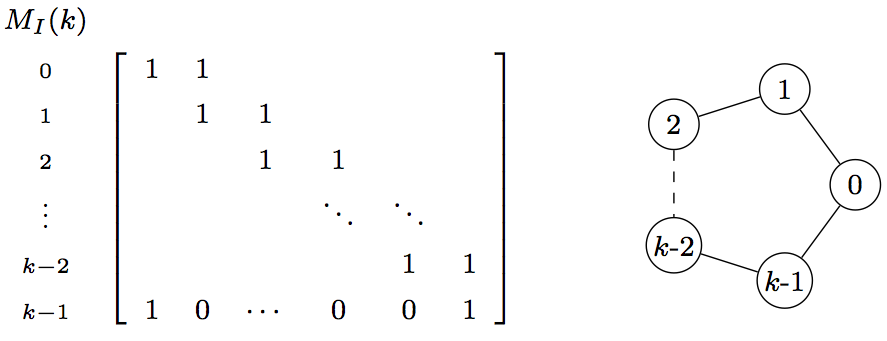}}\\
\subfloat[$G_{IV}(n)$, $n \geq 6$, $k \geq 3$]{\includegraphics[scale=.23]{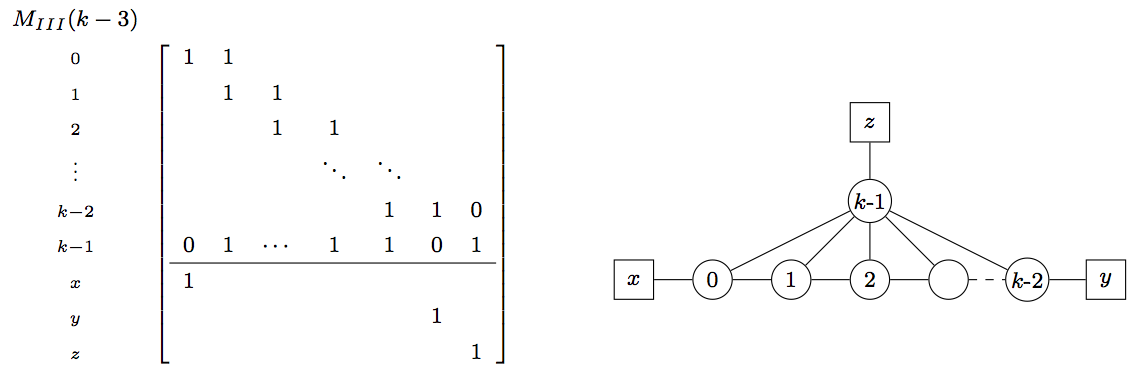}}\\

\subfloat[$G_V(n)$, $n \geq 6$, $k \geq 4$]{\includegraphics[scale=.21]{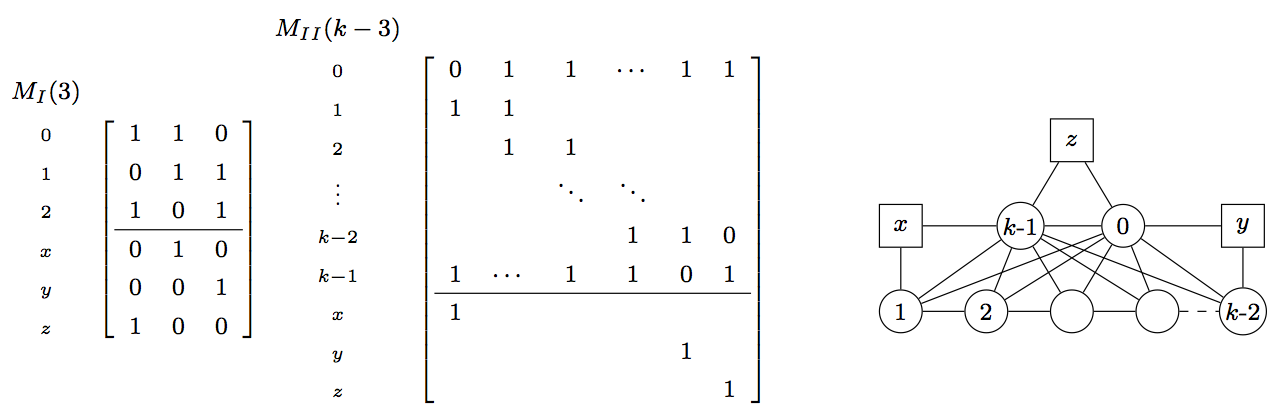}}
\caption{The relationship between clique matrices of the LB graphs and the Tucker matrices.
To the left of each graph is a corresponding clique matrix.
The square vertices of the graphs are those that occupy a single clique; these are the {\em simplicial}
vertices.
In each case, the submatrix obtained by excluding rows corresponding to simplicial vertices is 
a Tucker matrix, a fact observed by Tucker~\cite{Tuck72}.
$G_V(6)$ is a special case that has $M_I(3)$ as a submatrix of its clique matrix; 
$G_V(n)$ for $n \geq 7$ has $M_{II}(n-3)$ as a submatrix of its clique matrix.
}\label{fig:LBGraphsTucker}
\end{figure}

In this paper, we give a linear time bound for finding an induced LB
subgraph when a graph is not an interval graph.   As part of our algorithm, we
also give a linear-time $(O(size(M))$ bound for finding one of Tucker's 
submatrices in a matrix $M$ that does not have the consecutive-ones property. 
This latter problem was solved previously in $O(n*size(M))$ time in~\cite{Tamayo}.
An $O(\Delta^3 m^2n(m+n^2)$ bound 
for finding a Tucker submatrix of minimum size is given in~\cite{Dom10}, where $\Delta$ is
the maximum number of 1's in any row.

A {\em simplicial vertex} of a graph is a vertex that occupies a single clique;
it and its neighbors are the members of the clique.  In Figure~\ref{fig:LBGraphsTucker},
the square vertices are the simplicial vertices.
A {\em chord} on a path or cycle in a graph is an edge that is not on
the cycle or path, but both of whose endpoints lie on the cycle or path.
A {\em chordless cycle} in a graph is a cycle on four or more vertices that
has no chord, that is, it is an induced $G_{III}(n)$ for $n \geq 4$.
A {\em chordless path} is a path that has no chord.
A graph is {\em chordal} if it has no chordless cycle.  
Since $G_{III}(n)$ for $n \geq 4$ are forbidden induced subgraphs for
interval graphs, interval graphs are a subclass of chordal graphs.
The clique matrices
of graphs can have an exponential number of columns, but in the case of
chordal graphs, if $M$ is a clique matrix, $size(M) = O(n+m)$, see~\cite{Gol80}.  

A {\em certifying algorithm} is an algorithm that provides, with each output, a simple-to-check
proof that it has answered correctly~\cite{KMMSJournal,MMNSCert}.  An interval model gives a 
certificate that a graph is an interval graph, and an LB subgraph gives one 
if the graph is not an interval graph.  However, a certifying algorithm for
recognition of interval graphs was given previously
in~\cite{KMMSJournal}.  
The ability to give a consecutive-ones ordering or a Tucker 
submatrix in linear time gives a linear-time certifying algorithm for consecutive-ones
matrices, but one was given previously in~\cite{McCSODA04}.  
However, the previous certificates are neither minimal nor
uniquely characterized.  It is easy to obtain 
a minimal certificate of the form given in~\cite{KMMSJournal} from an LB subgraph found by
the algorithm we describe below, but not the other way around.  
The presentation in that paper gives the certificate in an especially easy format for
authenticating, and the LB subgraphs can easily be given in this format.
However, checking that it is also minimal would be more complicated
and unnecessary for certifying that the graph is not an interval graph.  

Therefore, interest in the algorithm of this paper will likely be motivated
by the theoretical importance of the LB subgraphs, rather than by certification.
Results such as those in this paper can have unanticipated algorithmic uses.
For example, the algorithm of Rose, Tarjan and Lueker~\cite{RTL:triangulated} 
recognizes whether a graph is a chordal graph, but it does not
return a chordless cycle if it is not.  The addendum
given by~\cite{TarjYan85} was a response to demand for an algorithm
to find chordless cycles in arbitrary non-chordal graphs.  Though
this is useful for certification, it does not appear to have been
the motivation for the addendum.
{\em Circular-arc graphs} are the intersection graphs of
arcs on a circle, and chordless cycles play a role in efficient recognition 
of circular-arc graphs~\cite{Tuck80, McCCircPaper, KaplanNussbaum11}.

The LB Subgraphs play a role in the characterization of related graph classes.
For any subclass of interval graphs, the LB subgraphs, or induced subgraphs of them, must be among 
the minimal forbidden induced subgraphs for the class.  For example,
an interval graph is {\em proper} if there exists an interval model where
no interval is a subset of another.  It is a {\em unit} interval graph if
there exists an interval model where all intervals have the same length.
These graph classes are the same, and they are a subclass of the interval
graphs.  Wegner showed that a graph is a proper
interval graph if and only if it does not have a chordless cycle,
the special case of $G_{IV}$ or $G_V$
for $n = 6$, or the {\em claw} ($K_{1,3}$) as an induced subgraph~\cite{Wegner67}.
Hell and Huang give an algorithm that produces one of them in linear time~\cite{HellUIG}.
The problem of finding a forbidden subgraph for this class reduces to finding an LB
subgraph:  Each of the LB graphs is either one of Wegner's forbidden subgraphs or contains
an obvious claw.  Therefore, if a graph $G$ is not an interval graph,
we can find one of Wegner's forbidden subgraphs in linear time.
If $G$ is an interval graph, it is trivial to find a claw in linear time
using an interval model of $G$ generated by the algorithm of~\cite{BL76}.
This approach has no obvious advantages over Hell and Huang's 
algorithm, but it illustrates that the application of our results is
not restricted to interval graphs.

This can also be true for properly overlapping classes of graphs or superclasses of
interval graphs.
The characterization of the general class of circular-arc graphs in terms of its minimal forbidden
induced subgraphs has remained elusive, and
LB subgraphs also figure heavily in partial characterizations,
for example, the recent characterizations of
those graphs that are {\em normal Helly} circular-arc graphs~\cite{LSSNormal13, GrippoSafe12}.
Adding an isolated vertex to any of the LB subgraphs that are circular-arc graphs gives
a minimal forbidden subgraph for the class of circular-arc graphs:  each of these must have a circular-arc
model that covers the entire circle, since it is not an interval graph, it is minimal with
respect to this property, and this precludes
an isolated vertex.  If a characterization of circular-arc graphs in terms of minimal forbidden induced subgraphs
is discovered, our results will likely be useful in an algorithm for finding one.

The paper also makes extensive use of generic techniques whose usefulness, to our knowledge, has not been
previously recognized.  An example is the extensive use of Lemma~\ref{lem:BLPrefix}
in various parts of the algorithm.  

Tucker recognized a relationship between
the LB graphs and the Tucker matrices:  a Tucker submatrix occurs in
the clique matrix of every LB graph.  The relationship between Tucker submatrices
of clique matrices of graphs and their induced LB subgraphs, however, has a much
richer structure than has been previously recognized, which is shown in the last
section of the paper.

An open question 
is whether a minimal, unique, especially simple, or otherwise interesting special case of the 
certificate from~\cite{McCSODA04} can be obtained by applying the algorithm of that paper to a Tucker 
submatrix obtained by the algorithm of the present paper.
Tucker submatrices may be useful in heuristics for finding large submatrices that have the consecutive-ones
property, small Tucker matrices, or identifying errors in biological data~\cite{StoyeWittler09, CCHS10}.  Our techniques
provide new tools for such heuristics.

The results of this paper appeared in preliminary form in~\cite{LindzeyMcConnell13}.

\section{Preliminaries}

Given a graph $G$, let $V$ denote the set of vertices and $E$ denote the set of edges.
Let $n$ denote $|V|$ and $m$ denote $|E|$.
If $\emptyset \subset X \subseteq V$, let $G[X]$ denote the subgraph induced by $X$.
By $G - x$, we denote $G[V \setminus \{x\}]$, and by $G - X$, we
denote $G[V \setminus X]$.

Given an ordered sequence $(a_1,a_2, \ldots, a_k)$, let a {\em prefix} of the
sequence denote a consecutive subsequence $(a_1, a_2, \ldots, a_i)$ for some $i$
such that $0 \leq i \leq k$.  (If $i = 0$, the prefix is empty.)
Similarly, let a {\em suffix} denote $(a_i, a_{i+1} \ldots a_k)$ for some $i$
such that $1 \leq i \leq k+1$.  

We treat the rows as {\em sets}, where
each row $R$ is the set of columns in which the row has a 1.
We will use calligraphic font for collections of sets, hence
for sets of rows.

For a set $\mR$ of rows and a set $C$ of columns of a matrix, let $\mR[C]$ denote
the restriction $\{R \cap C| R \in \mR\}$ of $\mR$ to columns in $C$.
Suppose $\mR$ is the set of rows of a consecutive-ones ordered matrix
and $(c_1, c_2, \ldots, c_m)$ is the ordering of the columns.
In linear time, we can find, for each row, the leftmost and rightmost column
in the row.  Let us call these the {\em left endpoint} and {\em right endpoint} of the row.
Every graph $G$ is the intersection graph of rows of the clique matrix,
since two vertices of $G$ are adjacent if and only if they are members of a common clique.
Therefore, a consecutive-ones ordering of the clique matrix
of a graph gives an interval model of the graph; the intervals are the
consecutive blocks of 1's in the rows.

That interval graphs are a subclass of the class of chordal graphs follows
from inclusion of the $G_{III}$'s among the LB subgraphs.  

When a graph is chordal, the problem of deciding
whether it is an interval graph reduces to the problem of deciding whether its clique
matrix has the consecutive-ones property.   
When $G$ is chordal, the algorithm of Rose, Tarjan and Lueker 
produces its maximal cliques, hence a sparse representation
of its clique matrix, in linear 
time~\cite{RTL:triangulated}.  Booth and Lueker gave
an algorithm for finding a consecutive-ones ordering of an arbitrary
matrix or else determining that it does not have the
consecutive-ones property, which yielded a linear-time bound
for interval graph recognition~\cite{BL76}.   

Booth and Lueker's algorithm actually achieves a stronger result, of which
we make extensive use in this paper:

\begin{lemma}\label{lem:BLPrefix}~\cite{BL76}
Given a matrix $M$, it takes $O(size(M))$ time to find
the maximal prefix of the rows of $M$ that has the consecutive-ones
property.
\end{lemma}

\begin{definition}\label{def:overlapGraph}
Let $\mR = \{R_1, R_2, \ldots, R_p\}$ be a subset of rows of a matrix.
Two rows {\em overlap} if their intersection is nonempty but
neither is a subset of the other.  
The {\em overlap graph} of $\mR$ is
the undirected graph whose vertices are the members of $\mR$, and where $R_i, R_j \in \mR$ are
adjacent if and only if $R_i$ and $R_j$ overlap.
By an {\em overlap component} of $\mR$, we denote the elements of $\mR$ that
make up a connected component of the overlap graph.
\end{definition}

\begin{definition}\label{def:Venn}
Suppose the overlap graph of a set $\mR_Q$ of rows is connected.
Then two columns of $M$ are in the 
same {\em Venn class} of $\mR_Q$ if they are elements of the same set of members of $\mR_Q$.
The {\em unconstrained Venn class} consists of those columns that are not in
any member of $\mR_Q$; all others are {\em constrained}.
\end{definition}

\begin{figure}[!ht]\label{fig:schematic}
\captionsetup[subfigure]{labelformat=empty}
\centering
\subfloat[]{\includegraphics[scale=.3]{Tucker5.png}}
\quad \quad \quad
\subfloat[]{\includegraphics[scale=.34]{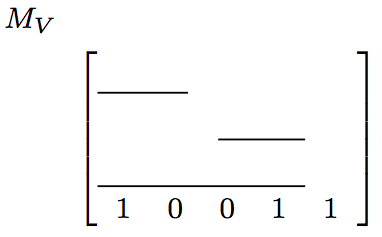}}
\caption{When some of the rows of a matrix are consecutive-ones
ordered, we will often represent the intervals occupied by
the 1's in these rows with line segments.
}
\end{figure}

\begin{figure}
\centerline{\includegraphics[scale=.25]{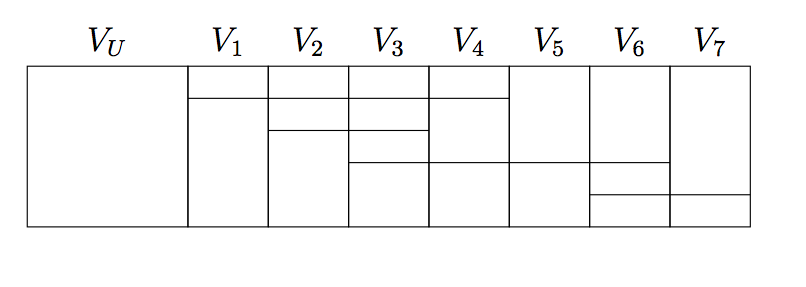}}
\caption{A set $\mR_Q$ of rows of a consecutive-ones matrix whose
overlap graph is connected, and its
Venn classes, $\{V_U, V_1, V_2, \ldots, V_7\}$.  Each element of $\{V_1, V_2, \ldots, V_7\}$
is consecutive, their union is consecutive, and the order of $V_1$ through $V_7$ is uniquely constrained
in any consecutive-ones ordering of $\mR_Q$, up to reversal.  These
are the {\em constrained} classes.  Columns in
the unconstrained class, $V_U$, can go
on either end of this sequence.
}\label{fig:Venn}
\end{figure}

\begin{lemma}\label{lem:Meidanis}~\cite{MeidanisPT98}
(See Figure~\ref{fig:Venn}.)
If the overlap graph of a set $\mR_Q$ of rows is connected,
then in all consecutive-ones orderings of $\mR_Q$,
each constrained Venn class is consecutive, the union of 
constrained Venn classes is consecutive, and the sequence of
constrained Venn classes in the ordering is invariant,
up to reversal.
\end{lemma}

In this paper, we will focus on sets of columns whose overlap
graph is a path.
Abusing notation slightly, when $\mP$ is the set of rows on a path in
the overlap graph, we will alternatingly treat it as an ordered sequence
or as an unordered set.  For example, a {\em prefix} of $\mP$
is a subpath of $\mP$ that contains its first row, and a {\em suffix}
is a subpath that contains its last row.  On the other hand, we
can treat it as an unordered collection of sets, as in the expression 
$\mP \cup \{Z\}$, where $Z$ is a row not in $\mP$.

\section{Breadth-first search on the overlap graph of the rows of a matrix, 
given a consecutive-ones ordering}\label{sect:BFS}

One step of our algorithm is to find a shortest path of the overlap graph
of the rows of a consecutive-ones ordered matrix.  
The difficulty in performing BFS on the overlap graph within the $O(size(M))$ time bound
is that the size of the overlap graph is not $O(size(M))$.  A simple example of this
is when $M$ has three columns, $n/2$ rows with 1's in the first and second column
and $n/2$ rows with 1's in the second and third column.  Then $size(M) = O(n)$
but its overlap graph is a complete bipartite graph with $n^2/4 = \Omega(n^2)$ edges.
The algorithm is given as Algorithm~\ref{alg:OverlapBFS}.

\begin{algorithm}
\KwData{A consecutive-ones ordered matrix $M$ and a starting row $S$.}
\KwResult{A BFS tree in the overlap graph on rows of $M$, rooted at $S$}
Let $Q$ be an empty queue of rows\;
Enqueue $S$ to $Q$\;
Let $\mR'$ be the rows of $M$ other than $S$ \;
\For {each column $c_i$}{
 Let $\mR_i$ be a doubly-linked list of rows in $\mR'$ whose right endpoint is in $c_i$,\\
 \quad \quad sorted in ascending order of left endpoint\;
   Let $\mL_i$ be a doubly-linked list of rows in $\mR'$ whose left endpoint is in $c_i$, \\
    \quad \quad sorted in descending order of right endpoint\;
}
\While {$Q$ is not empty}{
   Dequeue a row $R$\;
   Let $(c_j, c_{j+1}, \ldots, c_k)$ be the columns of $R$\;
   \For {$i = j+1$ to $k$}{
       $R' \longleftarrow$ the first row in $\mR_i$\;
       \While {the left endpoint of $R'$ is to the left of $c_j$}{
           Remove $R'$ from the front of $\mR_i$\;
           Let $h$ be the left endpoint of $R'$\;
           Remove $R'$ from $\mL_h$\;
           Assign $R$ as the parent of $R'$\;
           Enqueue $R'$ to $Q$\;
           $R' \longleftarrow$ the first row in $\mR_i$\;
       }
   }
   // Left-right mirror image of previous {\tt For} loop .. \;
   \For {$i = j$ to $k-1$}{
       $R' \longleftarrow$ the first row in $\mL_i$\;
       \While {the right endpoint of $R'$ is to the right of $c_k$}{
           Remove $R'$ from the front of $\mL_i$\;
           Let $h$ be the right endpoint of $R'$\;
           Remove $R'$ from $\mR_h$\;
           Assign $R$ as the parent of $R'$\;
           Enqueue $R'$ to $Q$\;
           $R' \longleftarrow$ the first row in $\mL_h$\;
       }
   }
}
\caption{{\tt OverlapBFS($M,S$)}\label{alg:OverlapBFS}}
\end{algorithm}

\begin{lemma}\label{lem:OverlapBFSCorrect}
The parent relation assigned by {\tt OverlapBFS$(M,S)$} (Algorithm~\ref{alg:OverlapBFS}) 
is a BFS tree, rooted at $S$, in the overlap graph of rows of $M$.
\end{lemma}
\begin{proof}
In BFS, each time a vertex $v$ is dequeued, the vertices that are
enqueued are those that have not previously
been enqueued and that are neighbors of $v$.  It suffices
to show that this is what {\tt OverlapBFS} does.

An invariant is that for each column $c_i$, the rows in $\mL_i$ are 
those rows that have not been enqueued and
whose left endpoints are in $c_i$, and the rows in $\mR_i$ are those
that have not been enqueued and whose right endpoints are in $c_i$.
This is true initially, and when a row is enqueued, it is removed
from the two lists $\mR_i$ and $\mL_h$ or $\mL_i$ and $\mR_h$
that it occupies, maintaining the invariant.

A row $R'$ that has not been enqueued
properly overlaps row $R = (c_j, c_{j+1}, \ldots, c_k)$ if and only if
one of of the following conditions applies:

\begin{enumerate}
\item The left endpoint of $R'$ is in $\{c_{j+1}, c_{j+2}, \ldots, c_k\}$ and its
right endpoints is to the right of $c_k$;

\item The right endpoint of $R'$
is in $\{c_j, c_{j+1}, \ldots, c_{k-1}\}$ and its left endpoint is to
the left of $c_j$.  
\end{enumerate}

The unenqueued rows meeting the first condition are
prefixes of $\{\mL_{j+1}, \mL_{j+1}, \ldots, \mL_k\}$ because of the way these
lists are sorted.  Similarly, the rows meeting the second condition are
prefixes of $\{\mR_j, \mR_{j+1}, \ldots, \mR_{k-1}\}$.  
The inner {\tt while} loops enqueue these prefixes.
Thus, when $R$ is dequeued,
the algorithm enqueues precisely those rows that have not previously been enqueued
and that properly overlap $R$, that is, the unenqueued neighbors of $R$ in the
overlap graph of rows of $M$.
\end{proof}

\begin{lemma}\label{lem:OverlapBFSTime}
{\tt OverlapBFS} can be implemented so that it can find the overlap components
of rows of a consecutive-ones ordered matrix $M$, and, for each component, a BFS tree, in $O(size(M))$ time.
\end{lemma}
\begin{proof}
In linear time, we may label each row of a consecutive-ones ordered matrix with its
left and right endpoints.  We may then radix sort the rows with right endpoint
as primary sort key and left endpoint as secondary sort key to obtain the
sorted lists $\mR_i$, in $O(n+m) = O(size(M))$ time.  Similarly, we may
obtain the lists $\mL_i$ in $O(size(M))$ time.  Since the lists are doubly-linked,
when a row is removed from $\mR_i$, it can be removed from the list $\mL_h$
corresponding to its left endpoint, $c_h$, in $O(1)$ time.

When $R$ is dequeued, the first inner while loop takes $O(1+k)$ time to process each list $\mR_i$,
where $k$ is the number of elements of $\mR_i$ that get enqueued by the step, and similarly
for $\mL_i$.  Processing $R$ therefore takes $O(|R|+q)$ time, where $q$ is the number
of vertices that get enqueued when $R$ is processed.  Over all rows in the set $\mR$ of rows
of the overlap component
that contains $S$, this takes time proportional to the sum of cardinalities
of members of $\mR$, since each member of $\mR$ is enqueued once.
Iteratively restarting it on a new unenqueued row $S$ each
time it finds an overlap component gives all overlap components in 
$O(size(M))$ time.
\end{proof}

\section{Finding a Tucker submatrix in a matrix that does not have the consecutive-ones property}

Our algorithm for finding a Tucker submatrix in a matrix $M$ that does not have the consecutive-ones
property is summarized in Algorithm~\ref{alg:TuckerSubmatrix}.

\begin{algorithm}
\KwData{A matrix $M$ that does not have the consecutive-ones property.}
\KwResult{The rows and columns of a Tucker submatrix have been returned.}
$M$ = {\tt TuckerRows($M,4$)} // Algorithm~\ref{alg:TuckerRows}\; 
\If {$M$ has $i \leq 4$ rows}{
    // The rows of $M$ are the rows of a Tucker submatrix \;
    Find the set $C$ of column vectors that make up a Tucker submatrix\;
}
\Else {
    // Every Tucker submatrix of $M$ contains the first five rows of $M$ \;
    $M \longleftarrow$  {\tt FindRows$(M)$}  // Algorithm~\ref{alg:FindRows} \;
    $C \longleftarrow$ {\tt FindColumns$(M)$} // Algorithm~\ref{alg:FindColumns} \;
}
return $M[C]$\;
\caption{\texttt{TuckerSubmatrix ($M$)}\label{alg:TuckerSubmatrix}}
\end{algorithm}

{\tt TuckerRows,} which it calls, takes as parameters a matrix $M$ that does not have the consecutive-ones
property and an integer $k$, which is 4 in this case.  It returns a matrix that
is an ordering of a subset of rows of $M$ and that does not have
the consecutive-ones property.  If it has at most $k=4$ rows,
these are the rows of every Tucker submatrix of its returned matrix,
and finding the columns in linear time is trivial.
Otherwise, every Tucker submatrix of this matrix contains the first
five rows, and possibly other rows.  

The next procedure, {\tt FindRows} takes as a parameter a matrix $M$
that fails to have the consecutive-ones property and such that every
Tucker submatrix of it contains the first five rows.  It returns 
a matrix $M'$ consisting of a subset of rows of $M$, where every 
Tucker submatrix of $M'$ contains all rows of $M'$.  Moreover, excluding
the last row of $M'$, its overlap graph is a path.

This matrix is then passed to {\tt FindColumns}, which has as a precondition
that its parameter satisfy the aforementioned conditions that $M'$ satisfies.
It returns a set $C$ of columns, such that $M'[C]$
is a Tucker matrix.  

\subsection{{\tt TuckerRows}}

In this section, we define {\tt TuckerRows} (Algorithm~\ref{alg:TuckerRows}),  which is called from 
Algorithm~\ref{alg:TuckerSubmatrix}.  It takes as a parameter
a matrix $M$ that does not have the consecutive-ones property, and a parameter $k$.
It returns a matrix $M'$ that is an ordering of a subset of rows of $M$, such that $M'$
does not have the consecutive-ones property.  If the number $i$ of rows
of $M'$ is at most $k$, then every Tucker matrix of $M'$ contains all $i$ rows.
Otherwise, every instance of a Tucker matrix of $M'$ contains the first $k+1$ rows
of $M'$, and possibly additional rows.  

It runs in $O(k*size(M))$ time.
We could find the rows of a Tucker submatrix in every case
by calling {\tt TuckerRows} with parameter $k=n$,
but that would take $O(n*size(M))$ time, which is not linear.
Since Algorithm~\ref{alg:TuckerSubmatrix} calls it with
with $k = 4$, this call takes $O(size(M))$ time.
If the returned matrix has more than five rows, then since every
Tucker submatrix of it contains at least five rows,
this excludes the possibility
of $M_{IV}$ or $M_V$, each of which has four rows.  This simplifies
the problem, and is one of the motivations for selecting value of 4
for the parameter $k$.

The strategy of the algorithm is based on the following lemma:

\begin{lemma}\label{lem:lastRowTucker}
If a set $\mR'$ of rows has the consecutive-ones property and $Z$ is a row such that
$\mR = \mR' \cup \{Z\}$ does not,
then $Z$ is one of the rows of every instance of a Tucker submatrix in $\mR$.
\end{lemma}
\begin{proof}
Suppose there exists an instance $M_T$ of a Tucker matrix whose rows
are contained in rows of $\mR'$.  Then $\mR'$ does not have the
consecutive-ones property, a contradiction.
\end{proof}

\begin{algorithm}
\KwData{A matrix $M$ that does not have the consecutive-ones property,  $k \geq 1$.}
\KwResult{Postconditions are given by Lemma~\ref{lem:TuckerRowsCorrect}}
$i \longleftarrow 1$\;
\While{$i \leq k$ and $M$ has at least $i$ rows}{
$(R_1, R_2, \cdots , R_r, Z) \longleftarrow$ the minimal prefix of rows of $M$ that does not have \\ \quad \quad \quad \quad \quad \quad \quad \quad \quad \quad \quad the consecutive-ones property (Lemma~\ref{lem:BLPrefix})\;
$M \longleftarrow  (Z, R_1, R_2, \cdots , R_r)$\;
$i \longleftarrow i + 1$\; 
}
return $M$;
\caption{\texttt{TuckerRows$(M, k)$}
\label{alg:TuckerRows} 
}
\end{algorithm}

\begin{figure}
\centerline{\includegraphics[scale=.25]{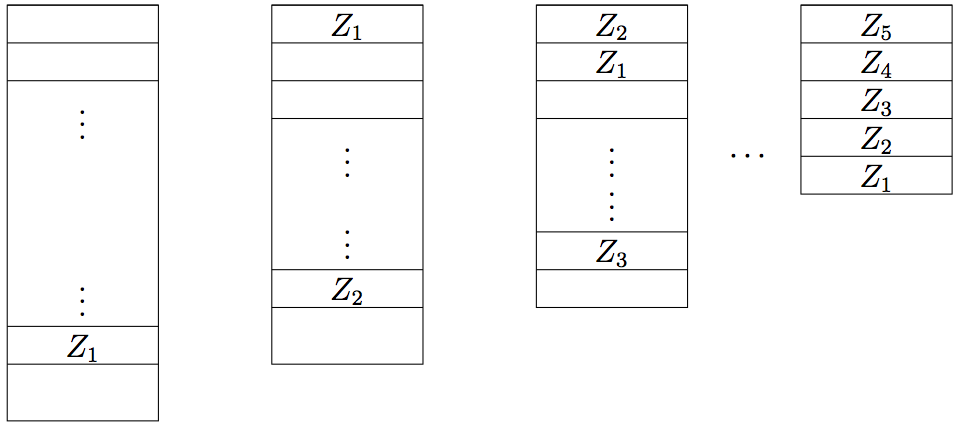}}
\caption{The first iteration of Algorithm~\ref{alg:TuckerRows} finds the minimal
prefix of the rows of $M$ that does not have the consecutive-ones property,
ending at some row $Z_1$.  After moving $Z_1$ to the
beginning of the matrix, the second iteration finds the minimal prefix that
does not have the consecutive-ones property, ending at some row $Z_2$.  Iterating
this operation, after $i$ iterations, we see by induction using 
Lemma~\ref{lem:lastRowTucker} that every instance of a Tucker matrix in the remaining
rows contains every row of $\{Z_i, Z_{i-1}, \ldots, Z_1\}$.
If that is all of the remaining rows, the algorithm can return them as the rows of a Tucker matrix.
If it halts after $k+1$ iterations, every Tucker submatrix
in the remaining rows of $M$ contains the first $k+1$ rows.
}\label{fig:fiveRows}
\end{figure}

\begin{lemma}\label{lem:TuckerRowsCorrect}
Suppose {\tt TuckerRows} (Algorithm~\ref{alg:TuckerRows}) is run with parameter $k$ and a matrix $M$ that does
not have the consecutive-ones property.  If the returned matrix $M'$ has at most
$k$ rows, then these are the rows of every Tucker submatrix in $M'$.
Otherwise, $M'$ fails to have the consecutive-ones property and every Tucker submatrix 
in $M'$ contains the first $k+1$ rows of $M'$.
\end{lemma}
\begin{proof}
By induction on $i$, $M$ does not have the consecutive-ones
property at the end of iteration $i$.
Also by induction on $i$, using Lemma~\ref{lem:lastRowTucker},
at the end of iteration $i$, either $M$ has at least $i$ rows
and every Tucker submatrix in $M$ 
contains the first $i$ rows of $M$, or else $M$ has
only $i-1$ rows and every Tucker submatrix in $M$ contains
these $i-1$ rows, in which case $M$ is returned before another iteration
takes place.
\end{proof}

\begin{lemma}\label{lem:TuckerRowsTime}
Algorithm~\ref{alg:TuckerRows} takes $O(k*size(M))$ time.
\end{lemma}
\begin{proof}
It has at most $k+1$ iterations of the loop, each of which takes
$O(size(M))$ time by Lemma~\ref{lem:BLPrefix}.
\end{proof}

\subsection{{\tt FindRows}}

Since Algorithm~\ref{alg:TuckerSubmatrix} only calls the remaining procedures
if {\tt TuckerRows} returns a submatrix with at least five rows,
we may assume the following henceforth,
by Lemma~\ref{lem:TuckerRowsCorrect}:

\begin{itemize}
\item $M$ is a matrix that does not have the consecutive-ones
property and where every Tucker submatrix contains the first five rows of $M$.
\end{itemize}

The purpose of {\tt FindRows} is to find the rows of a Tucker submatrix
in a matrix meeting this condition.
Since $M_{IV}$ and $M_V$ have only four rows, we may exclude them from
consideration.   
The only Tucker submatrices we need to consider henceforth are $M_I$,
$M_{II}$, and $M_{III}$.

\begin{proposition}\label{prop:cycle}
The overlap graphs of $M_I(k)$, $M_{II}(k)$, and $M_{III}(k)$ are simple cycles.
\end{proposition}

Let $\{Z_1, Z_2, \ldots, Z_5\}$ be the first five rows of $M$.
Let us choose $Z \in \{Z_1, Z_2, \ldots, Z_5\}$,
and let $\mR'$ be the remaining rows of $M$, excluding $Z$.
Since every instance of a Tucker submatrix contains $Z$, $\mR'$
has the consecutive-ones property, but $\mR' \cup \{Z\}$ does not.

Removal of one element from a chordless cycle gives
a chordless path.  Therefore, we seek a chordless path in the overlap
graph of $\mR'$ that has the consecutive-ones property, and such that when we add $Z$
to the rows on the path, they no longer have the consecutive-ones property.

\begin{definition}\label{def:suitable}
Let $\mR'$ be a set of rows of $M$ such that $\mR'$ has
the consecutive-ones property, but $\mR' \cup \{Z\}$ does not.
Rows $A,B \in \mR'$ are a {\em suitable pair  for $Z$}
if they are members of the same overlap component $\mR_Q$
of $\mR'$, each of $A$ and $B$ contains a 1 of $Z$, and in a consecutive-ones
ordering of $\mR_Q$, a 0 of $Z$ lies in between $A$ and $B$.  
(See Figure~\ref{fig:PP1P2}.)
\end{definition}

Our strategy is to find a suitable pair $\{A,B\}$ for some $Z$
and find a shortest, hence chordless, path
$\mP = (A = R_1, R_2, \ldots, R_k = B)$ between rows
$A$ and $B$ in the overlap graph.  $\mP$ must exist, since
$A$ and $B$ are members of the same overlap component $\mR_Q$.    
(See Figures~\ref{fig:PP1P2} and~\ref{fig:Path}.)
It is easy to see by Lemma~\ref{lem:Meidanis} that every consecutive-ones
ordering of $\mR_Q$ forces a 0 of $Z$ between two 1's,
and since these three columns are in distinct Venn classes
of $\mP$, this is true of the rows of $\mP$ also.  Therefore, $\mP$
has the consecutive-ones property but $\mP \cup \{Z\}$ does not, and
the overlap graph of $\mP$ is a path.

Unfortunately, the introduction of $Z$ may introduce chords in the overlap
graph between $Z$ and rows of $\mP$ other than its endpoints, $A$ and $B$.  
Therefore, Proposition~\ref{prop:cycle} does not imply
that  $\mP \cup \{Z\}$ is the set of rows of an instance
of a Tucker matrix.
We could find a smaller chordless cycle $\mP' \cup \{Z\}$ in the 
overlap graph of $\mP \cup \{Z\}$, but then we would run the
risk that $\mP' \cup \{Z\}$ would have the consecutive-ones property, defeating
our effort to find a minimal set of rows that does not have the consecutive-ones
property.  

Our solution is to show that if we find the minimal prefix $\mP_1$ of $\mP$ such that 
$\mP_1 \cup \{Z\}$ does not have the consecutive-ones property, and then
the minimal suffix $\mP_2$ of $\mP_1$ such that $\mP_2 \cup \{Z\}$ does not
have the consecutive-ones property, then $\mP_2 \cup \{Z\}$ is a minimal
set of rows that does not have the consecutive-ones property.  Therefore,
it must be the rows of a Tucker matrix.  

\begin{figure}
\centerline{\includegraphics[scale=.25]{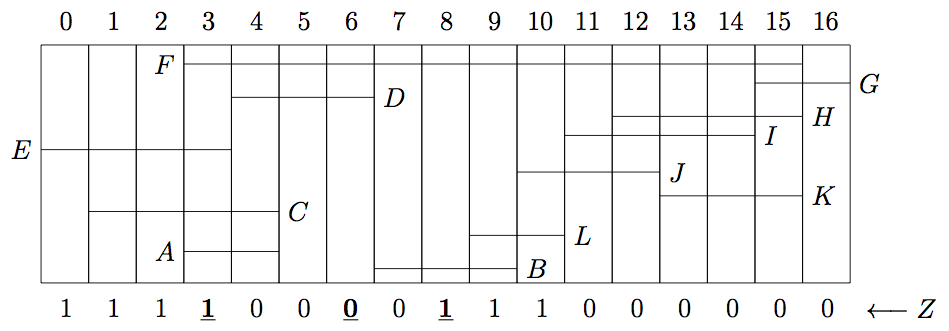}}
\caption{Example of finding a minimal set of rows that does not have
the consecutive-ones property.  The set $\mR_Q$ of rows, excluding $Z$,
has the consecutive-ones property, but $\mR_Q \cup \{Z\}$ does not.
Rows $A$ and $B$ are a suitable pair for $Z$, as shown 
by the boldface 1, 0, and 1.}\label{fig:PP1P2}
\end{figure}

{\bf Example:}  
{\em 
In Figure~\ref{fig:PP1P2},
$\mP = (A,E,F,G,H,J,L,B)$ is a shortest path from $A$ to $B$ in the overlap graph
of $\mR_Q$.  
In a consecutive-ones ordering of $\mP$ 
(Figure~\ref{fig:Path}),
the 0 in column 6 is forced to go between
the 1's in columns 3 and 8 because of where $A$ and $B$ must be placed in a consecutive-ones
ordering of $\mP$.  Therefore, $\mP \cup \{Z\}$ does not have the consecutive-ones property.

Next, $(A,E,F,G,H,J)$ is the smallest prefix $\mP_1$ of $\mP$ such that $\mP_1 \cup \{Z\}$
does not have the consecutive-ones property, and $(F,G,H,J)$ is the smallest
suffix $\mP_2$ of $\mP_1$ such that $\mP_2 \cup \{Z\}$ does not have the consecutive-ones
property.  $\mP_2 \cup \{Z\} = \{F,G,H,J,Z\}$ gives the rows of a Tucker matrix  
(Figure~\ref{fig:P1P2}). 
}

\begin{figure}
\centerline{\includegraphics[scale=.3]{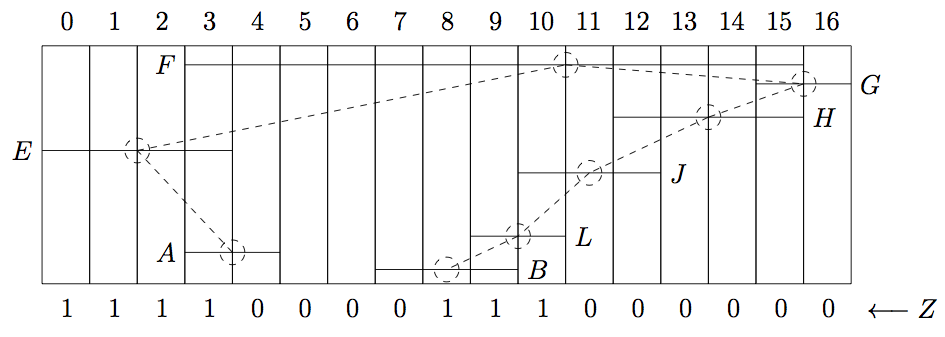}}
\caption{A shortest, hence chordless, path $\mP$ in the overlap graph of Figure~\ref{fig:PP1P2}
between $A$ and $B$.
}\label{fig:Path}
\end{figure}

\begin{figure}
\centerline{\includegraphics[scale=.27]{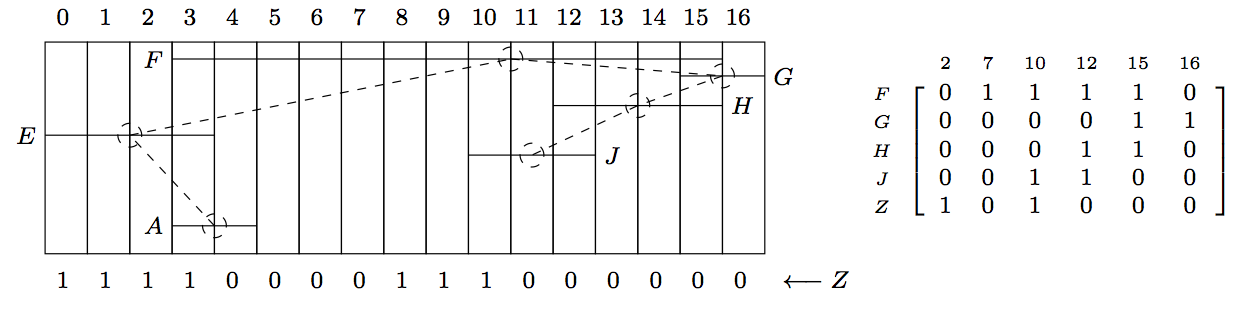}}
\centerline{\includegraphics[scale=.27]{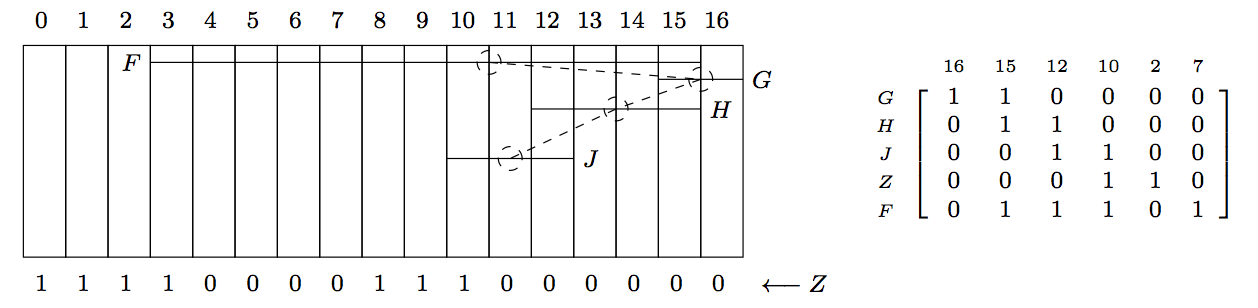}}
\caption{In Figure~\ref{fig:Path}, $(A,E,F,G,H,J)$ is the minimal prefix $\mP_1$ of $\mP$ such that
$\mP_1 \cup \{Z\}$ does not have the consecutive-ones property, and
$(F,G,H,J)$ is the minimal suffix $\mP_2$ of $\mP_1$ such that $\mP_2 \cup \{Z\}$
does not have the consecutive-ones property.  These can be found efficiently,
by Lemma~\ref{lem:BLPrefix}.  Then $\mP_2$
is a minimal subpath of $\mP$ whose union with $Z$ fails to have the consecutive-ones
property, hence $\mP_2 \cup \{Z\}$ gives the rows of a Tucker matrix.
The matrix at the top right is the submatrix we obtain after identifying
the columns (Algorithm~\ref{alg:FindColumns}, below),
and the permutation of rows and columns given at the bottom
right reveals that it is an instance of $M_{III}(5)$.}\label{fig:P1P2}
\end{figure}

A shortest path $(R_1, R_2, \ldots, R_k)$ between a suitable pair $A$ and $B$
can be found efficiently using {\tt OverlapBFS}.
Finding a minimal prefix whose union with $\{Z\}$
does not have the consecutive-ones property reduces to finding the
minimal prefix of $(Z, R_1, R_2, \ldots, R_k)$ that does not have the
consecutive-ones property.  
Finding a minimal suffix whose union with $\{Z\}$ does not have
the consecutive-ones property is solved similarly.  These problems can be
solved efficiently by Lemma~\ref{lem:BLPrefix}.
We give an algorithm for finding a suitable pair below.
The procedure is summarized as Algorithm~\ref{alg:FindRows}, where $\mP = (R_1, R_2, \ldots, R_k)$,
$\mP_1 = (R_1, R_2, \ldots, R_j)$, and $\mP_2 = (R_i, R_{i+1}, \ldots, R_j)$.

\begin{algorithm}
\KwData{A matrix $M$ that does not have the consecutive-ones property and such that every
Tucker submatrix contains the first five rows of $M$.}
\KwResult{A submatrix $M'$ of rows that does not have the consecutive-ones property, every Tucker submatrix of $M'$
contains all rows of $M'$, and, excluding the last row, the overlap graph of rows of $M'$ is a path.}
Let $(A,B,Z) \longleftarrow$ {\tt SuitablePair$(M)$} (Algorithm~\ref{alg:SuitablePair})\;
Let $\mR'$ be the rows of $M$ excluding $Z$\;
$(A = R_1, R_2, \cdots , R_k = B) \longleftarrow$  a shortest path from $A$ to $B$ in the overlap\\ \quad \quad  graph of $\mR'$ (Algorithm~\ref{alg:OverlapBFS})\; 

$(Z,R_1,R_2, \cdots ,R_j) \longleftarrow$ the minimal prefix of $(Z, R_1, R_2, \cdots , R_k)$ that does not\\ \quad \quad  have the consecutive-ones property (Lemma~\ref{lem:BLPrefix})\;
$(Z, R_j, R_{j-1}, \ldots, R_i) \longleftarrow$ the minimal prefix of $(Z, R_j, R_{j-1}, \cdots, R_1)$ that\\ \quad \quad does not have the consecutive-ones property (Lemma~\ref{lem:BLPrefix})\;
Return $(R_i, R_{i+1}, \cdots, R_j, Z)$\;
\caption{{\tt FindRows($M$)}}\label{alg:FindRows}
\end{algorithm}

\begin{lemma}\label{lem:FindRowsCorrect}
{\tt FindRows} returns a matrix $M'$ that does not have the consecutive-ones
property and where every Tucker submatrix in $M'$ contains all rows of $M'$.
If $p$ is the number of rows of $M'$, then the overlap graph of the first $p-1$
rows is a path.
\end{lemma}
\begin{proof}
In the proof of correctness of {\tt SuitablePair} (Algorithm~\ref{alg:SuitablePair}),
below, we show that $(A,B,Z)$ exists.
Let $\mP = (A = R_1, R_2, \ldots, R_k = B)$.  Since $A$ and $B$ lie in the same 
overlap component, $\mR_Q$, $\mP$ exists.  It is a shortest path, hence a chordless
path in the overlap graph.  The first $p-1$ rows of the returned matrix
is a subpath of $\mP$, so their overlap graph is a path.

Since $\mR_Q$ has the consecutive-ones property, so does $\mP$, and
any consecutive-ones ordering of $\mR_Q$ is a consecutive-ones
ordering of $\mP$.  Every consecutive-ones ordering
of $\mR_Q$ forces the 0 that lies between $A$ and $B$ to lie between the two
1's in $A$ and $B$, by Lemma~\ref{lem:Meidanis}.  This is true of at least one consecutive-ones
ordering of $\mP$.  The order of Venn classes of $\mP$ is
uniquely determined up to reversal, and the 0 that lies between $A$ and $B$
and the 1's in $A$ and $B$ are in different constrained Venn classes of $\mP$,
so this 0 lies between the two 1's in every consecutive-ones ordering of $\mP$.
$\mP \cup \{Z\}$ does not have the consecutive-ones property.

Since $\mP \cup \{Z\}$ does not have the consecutive-ones property, 
the shortest prefix $\mP_1$ of $\mP$ such that $\mP_1 \cup \{Z\}$ 
does not have the consecutive-ones property exists.
By the definition of $(Z, R_1, R_2, \ldots, R_j)$, 
$\mP_1 = (R_1, R_2, \ldots, R_j)$.  Since $\mP_1 \cup \{Z\}$
does not have the consecutive-ones property, 
the shortest suffix $\mP_2$ of $\mP_1$ such that 
$\mP_2 \cup \{Z\}$ does not have the consecutive-ones property exists.
By the definition of $(Z, R_j, R_{j-1}, \ldots, R_i)$,
$\mP_2 = (R_i, R_{i+1}, \ldots, R_j)$.

Suppose there is a proper subset $\mR'$
of the rows on $\mP_2$ such that $\mR' \cup \{Z\}$ does not have the
consecutive-ones property.  Since $\mP_2$ is a shortest path, it is a chordless path,
so $\mR'$ is a subpath of $\mP_2$ by Proposition~\ref{prop:cycle}.
Let $\mR_1 = (R_1, R_2, \ldots, R_{j-1})$.  This is the result of removing
the last row from $\mP_1$.
Let $\mR_2$ $= (R_{i+1},$ $ R_{i+2}, \ldots, R_j)$.  This is the result
of removing the first row from $\mP_2$.
By the minimality of $\mP_1$ and $\mP_2$, $\mR_1 \cup \{Z\}$ and $\mR_2 \cup \{Z\}$
have the consecutive-ones
property.  Since $\mR'$ is a proper subpath of $\mP_2$,
$\mR' \subseteq \mR_1$ or $\mR' \subseteq \mR_2$, so $\mR' \cup \{Z\}$ 
has the consecutive-ones property, contradicting our assumption that it does not.
Therefore, $\mP_2 \cup \{Z\}$ is a minimal set of rows that does not have
the consecutive-ones property.  It must be 
a minimal set of rows that contains an instance of a Tucker matrix.
\end{proof}

\subsection{Finding a suitable pair}

The procedure {\tt SuitablePair}, which is called from {\tt FindRows} (Algorithm~\ref{alg:FindRows}),
takes as a parameter a matrix $M$ that does not have the consecutive-ones
property and such that every instance of a Tucker submatrix in $M$ contains the first five rows.

\begin{lemma}\label{lem:corral}
Suppose a matrix $M$ fails to have the consecutive-ones property,
has no Tucker submatrix with fewer than five rows and every Tucker
matrix contains the rows $\{Z_1, Z_2, Z_3, Z_4, Z_5\}$.
For each $Z_i \in \{Z_1, Z_2, \ldots, Z_5\}$, let $\mR_i$ be the rows
of $M$, excluding $Z_i$.  For one of the five choices of $Z_i$,
$\mR_i$ has an overlap component with a suitable pair $A,B$ for $Z_i$.
\end{lemma}
\begin{proof}
(See Figure~\ref{fig:TuckRotations}.)
Let $M_T$ be an instance of  Tucker submatrix in $M$.  Since 
$M_T$ must have at least five rows, it is an instance of
$M_I(k)$, $M_{II}(k)$ or $M_{III}(k)$ for $k \geq 5$.  

For each choice of $Z_i$, $\mR_i$ has the consecutive-ones
property, since every Tucker matrix contains $Z_i$ as one of its
rows.  Let $\mR_T$ be the rows of $M_T$, excluding $Z_i$.
A consecutive-ones ordering of $\mR_i$ imposes a
consecutive-ones ordering on $\mR_T$.
Since the overlap graph of $\mR_T$ is connected,
by Proposition~\ref{prop:cycle}, the ordering of Venn
classes of $\mR_T$ is unique, up to reversal, by Lemma~\ref{lem:Meidanis}. 

Let the rows of $M_T$ be numbered as in Figure~\ref{fig:TuckMatrices}.
As shown in Figure~\ref{fig:TuckRotations},
unless $Z_i$ contains row $i \in \{0,1,k-2,k-1\}$ of $M_T$,
rows $i-1$ and $i+1$ of $M_T$ are a suitable pair for $Z_i$.
There are at most four of the five choices of $Z_i$ that can fail to
have a suitable pair in $\mR_i$.
\end{proof}

\begin{figure}[!ht]
  \captionsetup[subfigure]{labelformat=empty}
		\centering
\subfloat[$M_I$]{\includegraphics[scale=.2]{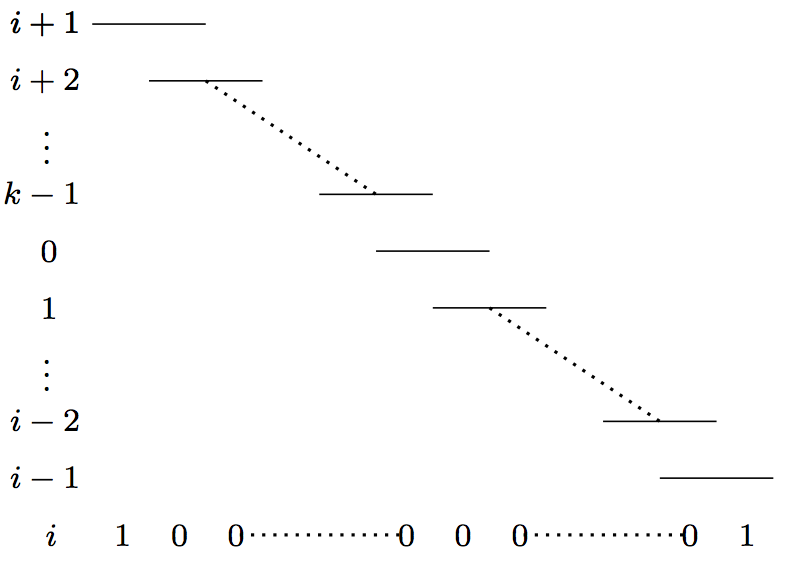}}\\
\subfloat[$M_{II}$]{\includegraphics[scale=.2]{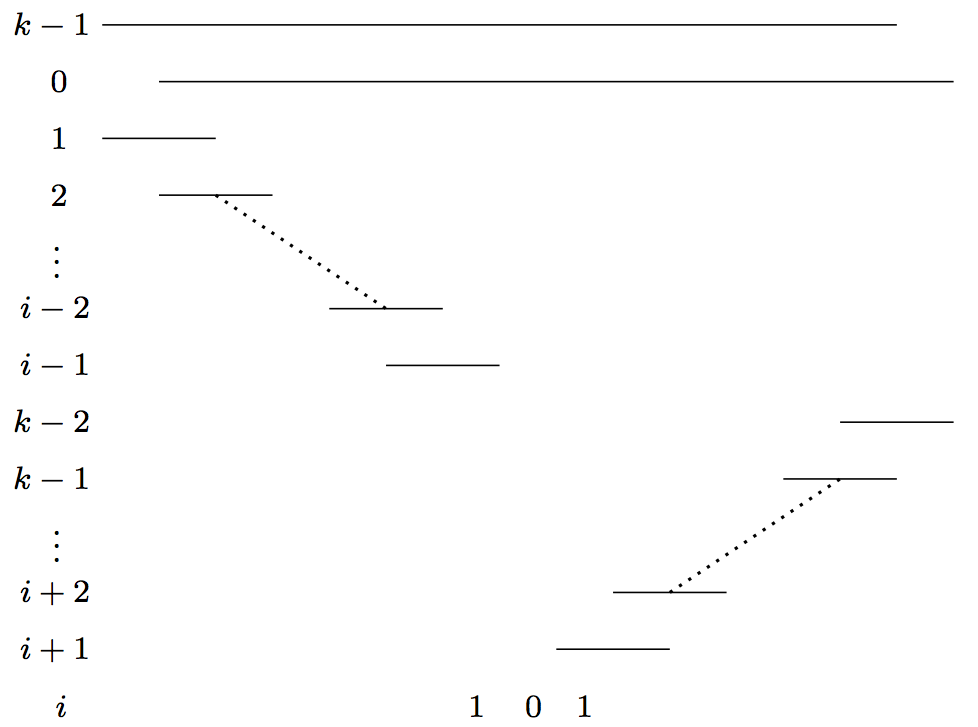}}\\
\subfloat[$M_{III}$]{\includegraphics[scale=.2]{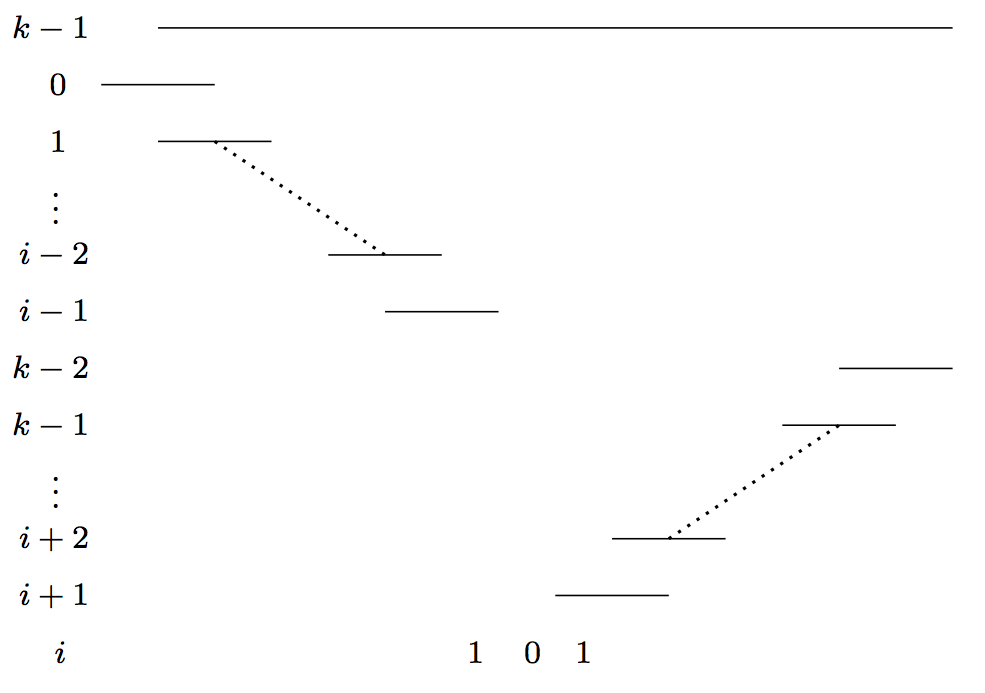}}
\caption{Consecutive-ones orderings of all but row $i$ for
any $i \in \{2, 3, \ldots, k-3\}$ in each of $M_I$, $M_{II}$
and $M_{III}$ gives rows $i-1$ and $i+1$ as a suitable pair for row
$i$ by Definition~\ref{def:suitable}}\label{fig:TuckRotations}
\end{figure}

\begin{lemma}\label{lem:findAB}
Suppose a set $\mR'$ of rows is consecutive-ones ordered, but
but $\mR' \cup \{Z\}$ does not have the consecutive-ones property.
Let $\mR_Q$ be the rows of an overlap component of $\mR'$.
Let $A'$ be a row with the leftmost right endpoint such
that $A'$ contains a column that has a 1 in row $Z$, and
let $B'$ be a row with the rightmost left endpoint such
that $B'$ contains a column that has a 1 of $Z$.
Then $A'$ and $B'$ are a suitable pair for $Z$ if 
$A'$ and $B'$ are disjoint, and a column that has a 0 in row $Z$ occurs in between $A'$ and $B'$.
Otherwise no suitable pair for $Z$ exists in $\mR_Q$.
\end{lemma}

\begin{proof}
If $A'$ and $B'$ satisfy the conditions, they are
a suitable pair for $Z$, by definition.
Conversely, suppose there exist rows $A, B \in \mR_Q$
that are a suitable pair for $Z$.  Suppose without loss of generality that $A$ is to the
left of $B$.   Then
$A'$ can be substituted for $A$ and $B'$ can be substituted for $B$, and
the 0 of row $Z$ that lies between $A$ and $B$ will also 
lie between $A'$ and $B'$, hence $\{A',B'\}$ is a suitable pair for $Z$.
\end{proof}

The procedure is given as Algorithm~\ref{alg:SuitablePair}.

\begin{algorithm}
\KwData{A matrix $M$ that does not have the consecutive-ones property,
and in which every Tucker submatrix contains the first five rows of $M$}
\KwResult{A row $Z$ and a suitable pair $A, B$ for $Z$ (Definition~\ref{def:suitable}) have been returned}
Let $(Z_1, Z_2, \ldots, Z_5)$ be the first five rows of $M$ \;
\For {$i = 1$ to 5}{
   Let $\mR_i$ be the rows of $M$, excluding $Z_i$\;
   Let $M'$ be a consecutive-ones ordering of $\mR_i$\;
   Use {\tt OverlapBFS} to find the overlap components of $M'$ (Algorithm~\ref{alg:OverlapBFS})\;
   \For {each overlap component $\mR_Q$ of $M'$}{
      Let $A'$ be the member of $\mR_Q$ with the leftmost right endpoint 
          among\\ \quad \quad rows that contain a 1 of $Z_i$ \;
      Let $B'$ be the member of $\mR_Q$ with the rightmost left endpoint
          among\\ \quad \quad  rows that contain a 1 of $Z_i$\;
      \If {a 0 of $Z_i$ occurs in between $A'$ and $B'$}{
          return $(A', B', Z_i)$\;
      }
   }
}
\caption{{\tt SuitablePair($M$)}\label{alg:SuitablePair}}
\end{algorithm}

\begin{lemma}\label{lem:SuitablePairCorrect}
{\tt SuitablePair$(M)$} returns a triple $(A,B,Z)$ such that $\{A,B\}$
are a suitable pair for $Z$.
\end{lemma}
\begin{proof}
By Lemma~\ref{lem:corral}, a suitable pair exists for one of the five
choices $Z_i \in \{Z_1, Z_2, \ldots, Z_5\}$.  By Lemma~\ref{lem:findAB},
the procedure finds a suitable pair $\{A',B'\}$ for $Z_i$ returning $(A',B',Z_i)$.
\end{proof}

\begin{lemma}\label{lem:SuitablePairTime}
{\tt SuitablePair$(M)$} takes $O(size(M))$ time.
\end{lemma}
\begin{proof}
On each iteration $i$, we may find a consecutive-ones ordering of $\mR_i$
in $O(size(M))$ time.  We may then 
label all columns of $M$ with the
value of $Z_i$ in the column in $O(size(M))$ time.
The calls to {\tt OverlapBFS} take $O(size(M))$ time by Lemma~\ref{lem:OverlapBFSTime}.
Since each column is labeled with the value of $Z$ that it contains,
it takes $O(|R|)$ time to find whether a set $R$ of columns contains a 1 in $Z$ or a 0 in $Z$.
It takes time proportional to the sum of cardinalities of rows in an overlap
component to find $A'$ and $B'$, and to determine whether the columns in
the interval between $A'$ and $B'$ contain a 0 in $Z$.  Over all connected
components, this check takes $O(size(M))$ time.  

Therefore, over the five iterations, the algorithm takes $O(5*size(M)) = O(size(M))$ time.
\end{proof}

\begin{lemma}\label{lem:FindRowsTime}
{\tt FindRows$(M)$} (Algorithm~\ref{alg:FindRows}) takes $O(size(M))$ time.
\end{lemma}
\begin{proof}
The call to {\tt SuitablePair} takes $O(size(M))$ time by Lemma~\ref{lem:SuitablePairTime}.
It takes $O(size(M))$ time to find a shortest path from $A$ to $B$ by Lemma~\ref{lem:OverlapBFSTime}.
It takes $O(size(M))$ time to find each of $(Z, R_1, R_2, \ldots, R_j)$
and $(Z, R_j, R_{j-1}, \ldots, R_i)$ by Lemma~\ref{lem:BLPrefix}.
\end{proof}

\subsection{{\tt FindColumns}}

In this section, we develop {\tt FindColumns} (Algorithm~\ref{alg:FindColumns}), which 
is called from Algorithm~\ref{alg:TuckerSubmatrix}.  Since it is called on
a matrix $M$ that is returned by {\tt FindRows}, we may assume henceforth
that $M$ satisfies the conditions of Lemma~\ref{lem:FindRowsCorrect}, that
is, that it fails to have the consecutive-ones property, every Tucker submatrix
of $M$ contains all rows of $M$, and, excluding the last row, $Z$, the overlap
graph of the rows of $M$ is a path, which we will denote by $\mP'$.

\begin{definition}\label{def:configs}
(See Figure~\ref{fig:configs}.)
Let $\mR_Q$ be a set of rows such that the overlap graph of $\mR_Q$
is connected, let $(c_1, c_2, \ldots, c_k)$
be the left-to-right order of columns in a consecutive-ones ordering of $\mR_Q$,
and let $Z$ be a row that is not in $\mR_Q$.
A {\em 1-0-1 configuration} is a sequence of three columns $(c_h, c_i, c_j)$
in three separate constrained Venn classes of $\mR_Q$ such that $h < i < j$,
$c_h$ contains a 1, $c_i$ contains a 0, and $c_j$ contains a 1 in row $Z$.
A {\em 0-1-0 configuration for $Z$} is defined in the same way, except that
$c_h$ contains a 0, $c_i$ contains a 1, and $c_j$ contains a 0 in row $Z$.
\end{definition}

\begin{figure}
\centering
\subfloat[]{\includegraphics[scale=.5]{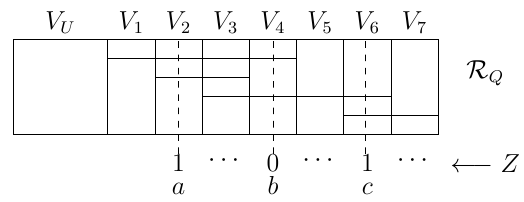}}\\
\subfloat[]{\includegraphics[scale=.5]{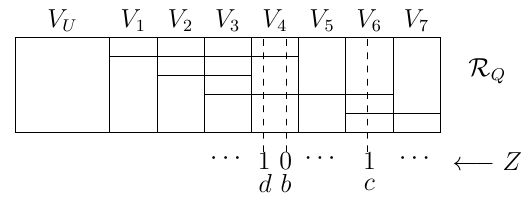}}\\
\subfloat[]{\includegraphics[scale=.5]{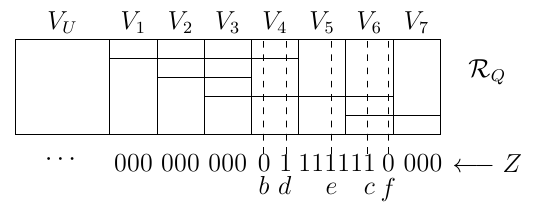}}
\caption{
(See Lemma~\ref{lem:violation2}).  $\mR_Q$ is a set of rows that has the consecutive-ones property and 
whose overlap graph is connected,  and $Z$ is
an additional row.  A 1-0-1 configuration for $Z$
is three columns in a consecutive-ones ordering of $\mR_Q$, such as $(a,b,c)$ 
in the top figure.
They must reside in three separate constrained  Venn classes of $\mR_Q$.
The 1-0-1 configuration proves that $\mR_Q \cup \{Z\}$
does not have the consecutive-ones property, since the 0 in $b$
is forced to be between the two 1's in $a$ and $c$ in every consecutive-ones ordering of $\mR_Q$,
by Lemma~\ref{lem:Meidanis}. 
Columns $(d,b,c)$ in the middle figure are not a 1-0-1 configuration, since they do not reside in three
separate Venn classes.   Since a Venn class can be reordered to give
a new consecutive-ones ordering, they do not exclude a consecutive-ones ordering
such as one that is consistent with the bottom figure.
A 0-1-0 configuration  is defined similarly, and $(b,e,f)$ is
an example in the bottom figure.  The 0-1-0 configuration
proves that $\mR_Q \cup \{Z\}$ does not have the consecutive-ones property if
if $V_U$ contains a 1 of Row $Z$:  the 1 in $e$ is separated from it by
0's in $b$ and $f$ in every consecutive-ones ordering of $\mR_Q$.
}\label{fig:configs}
\end{figure}

The sufficiency of the following is implicit in Booth and Lueker's
algorithm.  The necessity in the case where the overlap graph of the
rows is observed in~\cite{McCSODA04}.

\begin{lemma}\label{lem:violation2}
If $\mR_Q$ is a set of rows that has the consecutive-ones property and whose overlap
graph is connected, and $Z$ is a
row not in $\mR_Q$, then $\mR_Q \cup \{Z\}$ fails to have the consecutive-ones
property if and only if one of the following cases applies:

\begin{enumerate}
\item \label{firstCase} $\mR_Q$ has a 1-0-1 configuration for $Z$;
\item\label{secondCase}  $\mR_Q$ has a 0-1-0 configuration for $Z$ and $Z$
has a 1 in the unconstrained Venn class of $\mR_Q$.
\end{enumerate}
\end{lemma}

\begin{proof}
Let $(V_1, V_2, \ldots, V_k)$ be the ordering of the constrained Venn classes
in a consecutive-ones ordering of $\mR_Q$.  
Let $C$ be the columns of $M$ and
let $C_Q = \bigcup \mR_Q = \bigcup_{i=1}^k V_i$.
Note that $C \setminus C_Q$ is the unconstrained Venn class.
By Lemma~\ref{lem:Meidanis},
the ordering of $(V_1, V_2, \ldots, V_k)$ unique, up to reversal, $C_Q$ is consecutive,
and the columns within each Venn class can be reordered arbitrarily
to obtain a consecutive-ones ordering of $\mR_Q$.

Suppose one of the conditions applies for $Z$.
If the first condition applies, then, since the three columns of the
1-0-1 configuration lie in distinct members of $\{V_1, V_2, \ldots, V_k\}$,
the 0 of the configuration is forced between the two 1's in every
consecutive-ones ordering of $\mR_Q$, hence  
there can be no consecutive-ones ordering of $\mR_Q \cup \{Z\}$.
Similarly, if the second condition applies, the 1 of the 0-1-0
configuration is forced between the two zeros in every
consecutive-ones ordering of $\mR_Q$, so in every consecutive-ones
ordering of $\mR_Q$, it is separated
from the 1 in the unconstrained class by one of the two zeros in 
the configuration.  Therefore,
$(\mR_Q \cup \{Z\})$ does not have the consecutive-ones property.

Conversely, suppose neither of the two conditions applies.  If $C_Q$ contains
no 1's of $Z$, then the columns in $C \setminus C_Q$ can be freely ordered
to give a consecutive-ones ordering.  Assume henceforth that $C_Q$
contains a 1 of $Z$.

If $|\mR_Q| = 1$,
then $(\mR_Q \setminus Z, \mR_Q \cap Z, Z \setminus \mR_Q)$ is
a consecutive-ones ordering.
If $|\mR_Q| > 1$, then since two rows that are adjacent in the overlap
graph have three Venn classes, the number $k$ of Venn classes is at least 3.
To avoid a 1-0-1 configuration for $Z$, the Venn classes where
$Z$ has 1's must be a consecutive block $(V_i, V_{i+1}, \ldots, V_j)$ 
in $(V_1, V_2, \ldots, V_k)$,
and only $V_i$ and $V_j$
can have both 0's and 1's in $Z$.  If $V_i = V_j$, $V_i$ can be freely
ordered to give a consecutive-ones ordering of $(\mR_Q \cup \{Z\})[C_Q]$.
If $i < j$, the columns in $V_i$ can be freely
ordered so that its its 1's in $Z$ are consecutive with $V_{i+1}$
and the columns in $V_j$ can be freely ordered so that its 1's in $Z$
are consecutive with $V_{j-1}$, giving a consecutive-ones ordering
of $(\mR_Q \cup \{Z\})[C_Q]$.

If $Z$ has no 1 in $C \setminus C_Q$,
this is a consecutive-ones ordering of $\mR_Q \cup \{Z\}$.
Therefore, suppose there is no 1-0-1 configuration and $Z$ has a 1 in some column $c' \in C \setminus C_Q$.
Because of $c'$ and because the second condition of Lemma~\ref{lem:violation2} does not apply,
there is no 0-1-0 configuration for $Z$.  To avoid a 0-1-0 configuration, it must be that $(V_i, V_{i+1}, \ldots, V_j)$
is a prefix or a suffix of $(V_1, V_2, \ldots, V_k)$.
Without loss of generality, suppose it is a prefix, namely, $(V_1, V_2, \ldots V_j)$.
The columns in $V_j$ can be freely ordered so that those that contain 1's of $Z$ are
to the left of those that contain 0's, and the columns
containing $1$'s in unconstrained class can be placed so that they are consecutive
with $V_1$, giving a consecutive-ones ordering of $\mR_Q \cup \{Z\}$.
\end{proof}

\begin{lemma}\label{lem:minCols}
Let $M$ be a matrix that does not have the consecutive-ones property,
where every Tucker submatrix contains all rows of $M$, and where,
excluding the last row $Z$ of $M$, the overlap graph of rows of $M$
is a path, $\mP'$.
Let $C$ be be a subset of the columns of $M$.  Then
$M[C]$ is a Tucker submatrix of $M$ if and only if
$C$ is a minimal set of columns such that the following conditions apply:

\begin{enumerate}
\item  The overlap graph of $\mP'[C]$ is connected;
\item $(\mP' \cup \{Z\})[C]$ satisfies one of the conditions of Lemma~\ref{lem:violation2}
for $Z$.
\end{enumerate}
\end{lemma}
\begin{proof}
Since the overlap graph of $\mP'$ is a path, it is connected.  Since $\mP' \cup \{Z\}$
does not have the consecutive-ones property, it must satisfy
one of the conditions of Lemma~\ref{lem:violation2} for $Z$.
If $\mP'[C'']$ is connected and $(\mP' \cup \{Z\})[C'']$
satisfies one of the two conditions of the lemma, then 
$(\mP' \cup \{Z\})[C'']$ contains a Tucker submatrix since it does
not have the consecutive-ones property.
Conversely,
let $M_T$ be a Tucker matrix contained in $\mP' \cup \{Z\}$.   
By Lemma~\ref{lem:FindRowsCorrect}, $\mP' \cup \{Z\}$ is a minimal set of rows that contains a Tucker matrix,
so the rows of $M_T$ are $\mP' \cup \{Z\}$.
Let $C'$ be the columns of $M_T$.
Since the overlap graph of a Tucker matrix with at least five rows is a chordless cycle,
by Proposition~\ref{prop:cycle}, 
the overlap graph of $\mP'[C']$ is a chordless path.
Therefore, $C'$ satisfies the conditions of Lemma~\ref{lem:violation2},
and so does every set $C''$ of columns such that $C' \subset C''$.

Therefore, necessary and sufficient conditions for $(\mP' \cup \{Z\})[C'']$
to contain a Tucker submatrix is for
the overlap graph of $\mP'[C'']$ to be connected and for one of the conditions
of Lemma~\ref{lem:violation2} to apply for $Z$ in $(\mP' \cup \{Z\})[C'']$.
A minimal such set $C$ is a minimal set of columns that contains the columns
of a Tucker matrix in $\mP' \cup \{Z\}$.  Since $\mP' \cup \{Z\}$ is its set
of rows, $(\mP' \cup \{Z\})[C]$ is a Tucker matrix.
\end{proof}

Our algorithm for finding the columns of a Tucker matrix removes columns one at a time
from the set $C$ of columns,
except when doing so would undermine the requirements of Lemma~\ref{lem:minCols}
on rows of $(\mP' \cup \{Z\})[C]$.  
To obtain a linear time bound, we must describe an efficient
test of whether the removal of a column would undermine one of the requirements
of Lemma~\ref{lem:minCols}.

\begin{definition}\label{def:A}
Let $(R_1, R_2, \ldots, R_h)$ be the sequence rows of $\mP'$.
For $i \in \{1, 2, \ldots, h-1\}$,
let $\mA_i = \{R_i \setminus R_{i+1}, R_i \cap R_{i+1}, R_{i+1} \setminus R_i\}$.
Let $\mA = \bigcup_{i=1}^{h-1} \mA_i$.
\end{definition}

\begin{lemma}\label{lem:ASize}
$\sum_{i=1}^{h-1} (| R_i \setminus R_{i+1}| +  |R_i \cap R_{i+1}| + |R_{i+1} \setminus R_i|) = O(size(M))$.
\end{lemma}
\begin{proof}
For each $i \in \{1, 2, \ldots, h-1\}$, the members of 
$\{R_i \setminus R_{i+1}, R_i \cap R_{i+1}, R_{i+1} \setminus R_i\}$ are disjoint
and each is a subset of $R_i$ or of $R_{i+1}$.   The sum 
is at most twice the sum of cardinalities of members of $\{R_1, R_2, \ldots, R_h\}$.
\end{proof}

\begin{lemma}\label{lem:A}
Let $C \subseteq C_Q$.  The overlap graph of $\mP'[C]$ is
connected if and only if every element of $\mA$ contains an element
of $C$.
\end{lemma}
\begin{proof}
For $i \in \{1, 2, \ldots, h-1\}$,
$R_i[C]$ and $R_{i+1}[C]$ overlap if and only if each member of $\mA_i$
contains an element of $C$.
Since the overlap graph of $\mP' = (R_1, R_2, \ldots, R_h)$ is a chordless
path, the result follows.
\end{proof}

\begin{lemma}\label{lem:101Time}
If the overlap graph of $\mR_Q$ is connected, $\mR_Q$ is consecutive-ones
ordered, and $Z$ is a row not in $\mR_Q$,
it takes $O(size(M))$ time either to find a 1-0-1 configuration for $Z$
in $\mR_Q \cup \{Z\}$ or else to determine that no such configuration exists. 
Similarly, it takes $O(size(M))$ time either to find a 0-1-0 configuration for $Z$
in $\mR_Q \cup \{Z\}$, or else to determine that no such configuration exists. 
\end{lemma}
\begin{proof}
Let $(c_1, c_2, \ldots, c_k)$ be the consecutive-ones ordering of 
$\mR_Q$.  It takes $O(size(M))$ time to partition $(c_1, c_2, \ldots, c_k)$
into intervals corresponding to its constrained Venn classes; the
boundaries of the partition classes occur between consecutive columns such
that the first column is the right endpoint or the second column is a left endpoint
of a member of $\mR_Q$.

Label each column of $(c_1, c_2, \ldots, c_k)$ with the value of $Z$
in the column.
Suppose there exists a 1-0-1 configuration $(c_h, c_i, c_j)$ for $Z$.
If $c_p$ is the leftmost column of $(c_1, \ldots, c_k)$ that contains
a 1 of $Z$, then $(c_p, c_i, c_j)$ is a 1-0-1 configuration where $p \leq h$.  Find
the next column $c_q$ to the right of $c_p$ that resides in a different
Venn class from $c_p$ and contains a 0 of $Z$.  Because $(c_p, c_i, c_j)$
is a 1-0-1 configuration, $c_q$ exists and $q \leq i$.  Then $(c_p, c_q, c_j)$
is a 1-0-1 configuration.  Find the next column $c_r$ to the right
of $c_q$ that resides in a different Venn class from $c_q$ and contains a 1
of $Z$.  Because $(c_p, c_q, c_j)$ is a 1-0-1 configuration, $c_r$ exists.
Then $(c_p, c_q, c_r)$ is a 1-0-1 configuration.  The procedure
always succeeds in producing a 1-0-1 configuration if one exists.  Conversely, if
$c_p$, $c_q$ and $c_r$ exist, then a 1-0-1 configuration exists,
since $(c_p, c_q, c_r)$ is an example of one.  

The procedure
takes $O(k)$ time to find $c_p$ and scan rightward looking for
$c_q$ and $c_r$, once $(c_1, c_2, \ldots, c_k)$ and its partition
into constrained Venn classes is known.
By symmetry of the treatment of 0's and 1's of $Z$, and the second 
statement of the lemma also follows.
\end{proof}

The implementation of {\tt FindColumns} is based on these tests, and
is given as Algorithm~\ref{alg:FindColumns}.

\begin{algorithm} 
\KwData{A matrix $M$ that does not have the consecutive-ones property
and such that every Tucker matrix in $M$ contains all rows of $M$, and such
that, excluding the last row of $M$, the overlap graph of the rows is a path, $\mP'$.}
\KwResult{A Tucker submatrix of $\mP' \cup \{Z\}$}
    Let $C$ be the columns of $M$\;
    $C_Q \longleftarrow \bigcup \mP'$\;
    Let $(c_1, c_2, \cdots, c_k)$ be the left-to-right ordering
      of elements of $C_Q$ in a\\ \quad \quad  consecutive-ones ordering of $\mP'$\;
    Compute the members of $\mA$\;
    \If{there is a 1-0-1 configuration}  {
	    \tcp{The first condition of Lemma~\ref{lem:violation2} applies}
       $(c_{p(1)}, c_{p(2)}, c_{p(3)}) \longleftarrow$ a 1-0-1 configuration for $Z$\;
       $C \longleftarrow C_Q$\;
    }\Else {
	    \tcp{The second condition of Lemma~\ref{lem:violation2} applies}
       $C \longleftarrow C_Q \cup \{c'\}$\;
       $(c_{p(1)}, c_{p(2)}, c_{p(3)}) \longleftarrow$ a 0-1-0 configuration for $Z$\;
    }
    \For{$i \in (1,2, \ldots, k)$}{
        \If{ $c_i \not\in \{c_{p(1)}, c_{p(2)}, c_{p(3)}\}$ and
             $c_i$ is not the only element of $C$ in any member of $\mA$} {
               $C \longleftarrow C \setminus \{c_i\}$\; 
         }
     }
     \If{$(c_{p(1)}, c_{p(2)}, c_{p(3)})$ is a 1-0-1 configuration for $Z$}{
         \For{$i = 1$ to 3}{
             \If{$c_{p(i)}$ is not the only element of $C$ in any member of $\mA$ and $C \setminus \{c_{p(i)}\}$ has 
                  a 1-0-1 configuration for $Z$}{
                  $C \longrightarrow C \setminus \{c_{p(i)}\}$\;
             }
         }
     }
     \Else{
         \For{$i = 1$ to 3}{
             \If{$c_{p(i)}$ is not the only element of $C$ in any member of $\mA$ and $C \setminus \{c_{p(i)}\}$ has 
                  a 0-1-0 configuration}{
                  $C \longrightarrow C \setminus \{c_{p(i)}\}$\;
             } 
         }
     }
 
return $(\mP' \cup \{Z\})[C]$\;
\caption{\texttt{FindColumns$(M)$}}\label{alg:FindColumns}

\end{algorithm}

\begin{lemma}\label{lem:FindColumnsCorrect}
{\tt FindColumns$(M)$} returns a Tucker submatrix of $M$.
\end{lemma}
\begin{proof}
The initial matrix contains a Tucker submatrix.  
By Lemma~\ref{lem:A}, the tests applied before a column is removed ensure that the conditions
of Lemma~\ref{lem:minCols} continue to be satisfied by the final
submatrix $(\mP' \cup \{Z\})[C]$ returned by the procedure.  
Therefore, $(\mP' \cup \{Z\}[C]$ contains a Tucker submatrix.

We now show that $C$ is minimal with respect to this property.
If the overlap graph of $\mP'[C \setminus \{c\}]$ is not connected,
then $(\mP' \cup \{Z\})[C \setminus \{c\}]$ does not satisfy the conditions of
Lemma~\ref{lem:minCols}.  Thus, all columns of $C \cap C_Q \setminus \{c_{p(1)}, c_{p(2)}, c_{p(3)}\}$ 
are necessary for $(\mP' \cup \{Z\})[C]$
to satisfy the conditions of Lemma~\ref{lem:minCols}.  

If $(c_{p(1)}, c_{p(2)}, c_{p(3)})$ is a 1-0-1 configuration, due to the absence
of any element of the unconstrained class in $C$, after the first {\tt if} statement,
only condition 1 of
Lemma~\ref{lem:violation2} is satisfied, so each element
$\{c_{p(1)}, c_{p(2)}, c_{p(3)}\}$ that is retained in $C$ is necessary
for $(\mP' \cup \{Z\})[C]$ to satisfy the requirements of Lemma~\ref{lem:minCols},
either because its removal would undermine all 1-0-1 configurations or
would undermine the connectivity of the overlap graph of $\mP'$.

If $(c_{p(1)}, c_{p(2)}, c_{p(3)})$ is a 0-1-0 configuration, it is because
it is because condition 1 of Lemma~\ref{lem:violation2}
is not satisfied.  Removal of $c'$ from $C$ would undermine 
condition 2 of Lemma~\ref{lem:violation2}, hence the conditions
of Lemma~\ref{lem:minCols}.  Each element of
$\{c_{p(1)}, c_{p(2)}, c_{p(3)}\}$ that is retained in $C$ is necessary
for $(\mP' \cup \{Z\})[C]$ to satisfy the requirements of Lemma~\ref{lem:minCols},
either because its removal would undermine all 0-1-0 configurations or
would undermine the connectivity of the overlap graph of $\mP'$.

By Lemma~\ref{lem:minCols}, the returned submatrix, $(\mP' \cup \{Z\})[C]$, is a Tucker
submatrix.
\end{proof}

\begin{lemma}\label{lem:FindColumnsTime}
{\tt FindColumns$(M)$} can be implemented to take $O(size(M))$ time.
\end{lemma}
\begin{proof}
A consecutive-ones ordering of $\mP'$
takes $O(size(M))$ time by~\cite{BL76}.

By Lemma~\ref{lem:ASize}, it takes $O(size(M))$ time to list the
members of each instance of $R_i \setminus R_{i+1}$, $R_i \cap R_{i+1}$,
and $R_{i+1} \setminus R_i$ for $i \in \{1, 2, \ldots, h-1\}$.
This gives the members of $\mA$, some of them possibly more than once.
Let $\mL$ be this collection of lists.
Give each column a list of members of $\mL$
to which it belongs.  Initialize
a {\em cardinality counter} on each list in $\mL$, indicating the number of columns of $C$ that it contains.
These operations take $O(size(M))$ time using elementary methods, given
that the sum of lengths of the lists in $\mL$ is $O(size(M))$.

In the first {\tt for} loop, each column $c_i$ is tested to see whether it
one of the lists of $\mL$ that contains it has a cardinality counter of 1.
If not, $c_i$ is removed from $C$ and the cardinality counters of the
lists of $\mL$ that contain it are decremented.  The total time for
this loop is bounded by summing, over all $c_i \in \{c_1, c_2, \ldots, c_k\}$,
the number of members of $\mL$ that contain $c_i$.
This is just the sum of cardinalities of lists in $\mL$,
hence $O(size(M))$.

The procedure needs to test for or find a 1-0-1 or a 0-1-0 configuration on at most
five occasions:  two when it initially determines whether there is a 1-0-1 
or 0-1-0 configuration, thereby finding $(c_{p(1)}, c_{p(2)}, c_{p(3})$, and three
in one of the last two {\tt for} loops, when it tests whether removal of 
$c_{p(1)}$, $c_{p(2)}$, and $c_{p(3)}$ leave a 1-0-1 or 0-1-0 configuration.
Each of these occasions requires $O(size(M))$ time by Lemma~\ref{lem:101Time}.
\end{proof}

{\bf Example:}
{\em
We give an illustration of how Algorithm~\ref{alg:FindColumns} works
on the example of Figure~\ref{fig:P1P2} (bottom).
$\mA$ consists of the following sets: $\{3,4, \ldots, 14\}$, $\{15\}$ and $\{16\}$,
which ensure the overlap relation of $F$ and $G$, $\{16\}$, $\{15\}$, $\{12,13,14\}$,
which ensure the overlap relation of $G$ and $H$, $\{13,14,15\}$, $\{12\}$ and 
$\{10,11\}$, which ensure the overlap relation of $H$ and $J$.  Some of these
sets are redundant, but there is no need to detect this.  
We put a cardinality counter on each of these sets, and decrement it whenever
a column in the set is deleted.  At least one column
from each of these sets must be retained to maintain the connectivity of 
the overlap graph of $\mP'$.

Initially, $C = \{0, 1, \ldots, 16\}$, and $\mC_Q = \{3,4, \ldots, 16\}$.  There
is no 1-0-1 configuration, so the algorithm retains a column $c'$ in the unconstrained
class that contains a 1 in row $Z$.  Suppose it selects $c' = 2$.
It eliminates the remaining columns $0$ and $1$ of the unconstrained class
from $C$.  It finds the 0-1-0 configuration $(c_{p(1)}, c_{p(2)}, c_{p(3)}) = (4, 10, 12)$.  

The first {\tt for} loop eliminates all but columns $\{2, 4,10,12, 15,16\}$; 2 is skipped
because it is a member of the unconstrained Venn class, 4, 10, and 12
are skipped because they are $c_{p(1)}, c_{p(2)}$ and $c_{p(3)}$, and 15
and 16 are retained because they are the only remaining elements of some member of $\mA$ when they are reached.
The final {\tt for} loop determines that elimination of 4 or 10 would undermine
all 0-1-0 configurations, and elimination of 12 would remove the last remaining
element of the member $\{12\}$ of $\mA$.

The resulting instance of a Tucker matrix is that depicted on the
righthand side of Figure~\ref{fig:P1P2}.
}

That the algorithm is incomplete without the last two {\tt for} loops is illustrated
by the following example.  Let
$R_1 = \{c_1, c_2\}$, $R_2 = \{c_2, c_3, c_4\}$, $R_3 = \{c_4, c_5\}$,
$Z = \{c_1, c_2, c_5\}$, 
and $(c_1, c_2, c_3, c_4, c_5)$ be a consecutive-ones ordering
of $\mR_Q = \{R_1, R_2, R_3\}$.
Let $(c_{p(1)}, c_{p(2)}, c_{p(3)}) = (c_1, c_3, c_5)$.
None of the columns is removed by the first {\tt for} loop, since
each of $c_2$ and $c_4$ is the sole element of a member of $\mA$.   However,
$c_{p(2)}$ can be eliminated without undermining Lemma~\ref{lem:minCols}.
Although $c_{p(1)}$, $c_{p(2)}$, and $c_{p(3)}$ are in separate Venn classes,
there is no member of $\mA$ that contains $c_{p(2)}$
and excludes both $c_{p(1)}$ and $c_{p(3)}$.  The second {\tt for} loop
is required to eliminate it.  The third {\tt for} loop handles analogous situations
when $(c_{p(1)}, c_{p(2)}, c_{p(3)})$ is a 0-1-0 configuration.

\begin{theorem}\label{thm:TuckerTime}
It takes $O(size(M))$ time to find an instance of a Tucker submatrix in any
matrix that does not have the consecutive-ones property.
\end{theorem}
\begin{proof}
The calls to {\tt TuckerRows}, {\tt FindRows} and {\tt FindColumns} take $O(size(M))$ time
by Lemmas~\ref{lem:TuckerRowsTime},~\ref{lem:FindRowsTime} and~\ref{lem:FindColumnsTime}.
This gives the result if {\tt TuckerRows} returns a matrix with greater than four rows.
If {\tt TuckerRows} returns a matrix with $i \leq 4$ rows, then we cannot use {\tt FindColumns}
to find the columns of a Tucker submatrix.  However, by Lemma~\ref{lem:TuckerRowsCorrect},
every Tucker submatrix contains all $i$ rows.  One way to find one is to
generate all $i! \leq 24 = O(1)$ orderings of rows, and for each, to check
for the existence of the column vectors of length $i$, as
depicted in Figure~\ref{fig:TuckMatrices}.  One of the matrices has been found if all
of its depicted column vectors are found in one of the orderings.
This takes $O(m) = O(size(M))$ time.
\end{proof}

\section{Finding a Lekkerkerker-Boland Subgraph}

An {\em asteroidal triple} (AT) in a graph $G$ is a set of three vertices $\{x,y,z\}$
such that there is a path from $y$ to $z$ in $G  - N[x]$, a path from
$x$ to $z$ in $G - N[y]$, and a path from $x$ to $y$ in $G - N[z]$.  Lekkerkerker
and Boland showed that a chordal graph is an interval graph if and only if
it has no AT.
The AT's in the chordal LB subgraphs are those that are indicated by square vertices
in Figure~\ref{fig:LBGraphsTucker}.

In~\cite{KMMSJournal}, an algorithm is
given that finds an AT in a non-interval chordal graph in $O(n+m)$ time.  However,
it does not follow from this result that there is a linear-time algorithm for
finding an LB subgraph.  One reason is that the algorithm of that paper 
can produce AT's that are not the AT of any induced LB subgraph.  Some of the
difficulties posed by this and other pitfalls are explained
below.

Let $M$ be the clique matrix of a graph $G$.  For notational
convenience, we will treat the rows of $M$ as interchangeable with the
corresponding vertices of $G$.  This allows us to refer to the subgraph
of $G$ induced by a set $X$ of rows of $M$, for example.  Since
$G$ is the intersection graph of the rows of its clique
matrix, $G[X]$ is 
the intersection graph of the rows in $X$.  We will also treat
the columns interchangeably with the cliques they represent.
This allows us to refer to the intersection of two columns, for example.

Given that the clique matrices of chordal graphs have a consecutive-ones
ordering if and only if they are interval graphs, it follows
that the clique matrix of every non-interval chordal graph must
contain a Tucker matrix.
As observed above, the Tucker matrices other than $M_I(k)$ for
$k > 3$ are obtained by deleting the simplicial (square) vertices in Figure~\ref{fig:LBGraphsTucker}
from the clique matrices of chordal LB graphs.
Each of these vertices is simplicial, so each is a member of
single column of the clique matrix.
Let the {\em incomplete columns} of the Tucker matrices be
the three columns from which these simplicial vertices are
removed.  Other than for $M_{III}(k)$ for $k > 3$, this uniquely
defines the three incomplete columns for every Tucker matrix.
By {\em completing a Tucker matrix}, let us denote the inverse operation,
namely, for each incomplete column, adding a row that contains a 1
in that column and in no other column of the Tucker matrix.

One could mistakenly believe that, since an instance of a Tucker
submatrix in the clique matrix of $G$ proves that $G$ is not
an interval graph, and an induced LB subgraph proves the same thing,
all that is required to complete an induced LB subgraph of $G$
is to find the rows that complete the matrix.

The first pitfall is that an instance of a Tucker submatrix
cannot always be completed using rows of the clique matrix of $G$.
This illustrated by Figure~\ref{fig:LBExtension}.  

\begin{figure}
\centerline{\includegraphics[scale=.35]{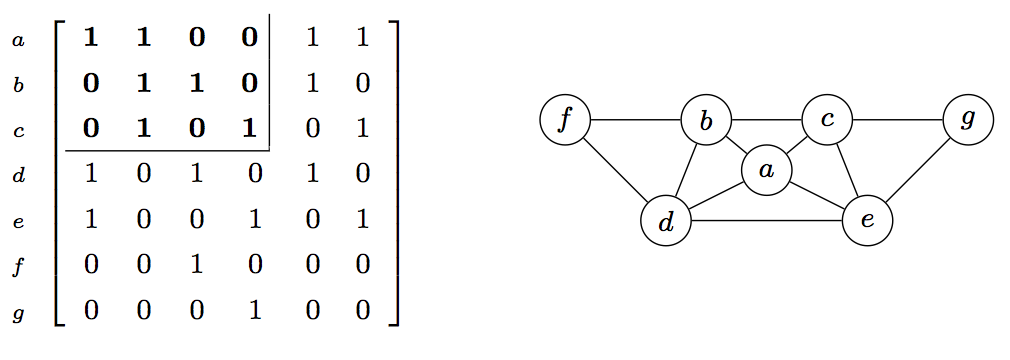}}
\caption{A Tucker submatrix of a clique matrix cannot always be completed
using rows of the clique matrix.  On the left is a clique matrix
of the graph on the right.  In the upper left is an instance
of a Tucker submatrix, $M_{III}(3)$.  
Its first, third and fourth columns are its incomplete
columns.  However, the first column cannot be completed
using rows of the clique matrix, since every row that
has a 1 in this column also has a 1 in another column
of the submatrix.  Its rows, $\{a,b,c\}$, are not vertices in 
any instance of an LB subgraph, since the only LB subgraph of $G$ is 
the chordless cycle $(b,c,e,d)$.}\label{fig:LBExtension}.
\end{figure}

Fortunately, a Tucker submatrix of the clique matrix of $G$
can always be completed if $G$ is chordal, as we show below.
However, it still does not follow that 
completing the Tucker matrix in this way yields an instance of an LB subgraph.
The problem is that that the intersection graph of rows of the 
instance of the Tucker matrix may not
faithfully represent the subgraph of $G$ that these rows induce.

Figure~\ref{fig:M4M5} illustrates this pitfall.
Suppose the depicted instance of $M_{IV}$ is found to be a submatrix
of the clique matrix $M$ of $G$.  Without loss of generality,
suppose that the depicted ordering of $M_{IV}$ is consistent
with the ordering of the clique matrix, $M$, and that the 1's in each of rows $0$, $1$ and $2$ are consecutive
in $M$.  If the intersection graph of these rows is the depicted instance $M_{IV}$, then 
the completion of this Tucker matrix yields the first graph
to its right, which is $G_I$.
However, it may be that rows 0 and 1 or that rows 1 and 2
intersect in one of the columns that do not form part
of the submatrix.  In $G$, the pairs $\{0,1\}$ and $\{1,2\}$
might be adjacent.

\begin{figure}[!ht]
  \captionsetup[subfigure]{labelformat=empty}
\centering
\subfloat[]{\includegraphics[scale=.3]{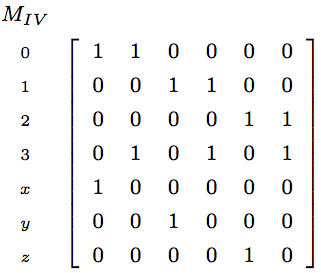}} \\
\subfloat[$G_I$]{\includegraphics[scale=.29]{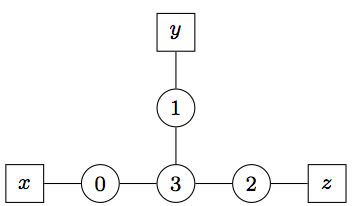}} \quad
\subfloat[$\times$]{\includegraphics[scale=.29]{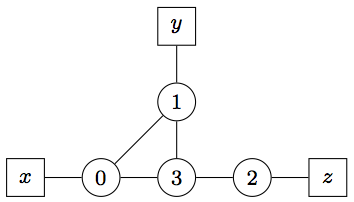}} \quad
\subfloat[$\times$]{\includegraphics[scale=.29]{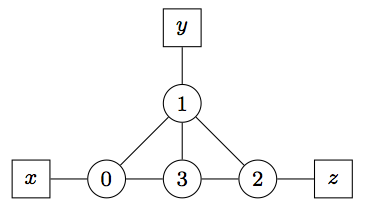}} \\
\subfloat[]{\includegraphics[scale=.3]{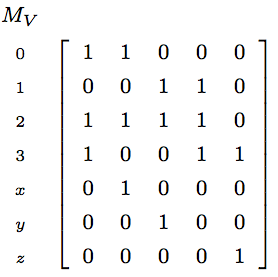}}\\
\subfloat[$G_{II}$]{\includegraphics[scale=.3]{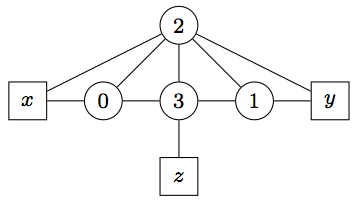}} \quad \quad
\subfloat[$\times$]{\includegraphics[scale=.26]{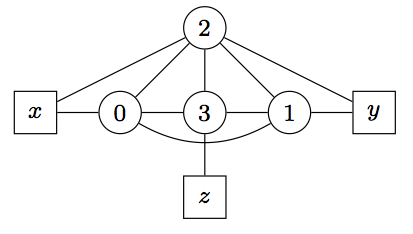}}

\caption{Even when a Tucker submatrix of a clique matrix can
be completed using rows of the clique matrix $M$ of a graph $G$,
this completion does not necessarily give the
vertices of an LB subgraph.  The intersection
graph of rows of the completed submatrix may not
accurately reflect the subgraph  of $G$ that they induce.
Rows that do not intersect in the submatrix may intersect in $M$.
The graphs below each matrix are the possible subgraphs induced
by the rows in the completion of the Tucker submatrices
$M_{IV}$ and $M_V$.
The middle graph below $M_{IV}$ is not an LB subgraph,
though it contains one, a $G_{IV}(6)$ induced by $\{x,y,0,1,2,3\}$.
The second graph below $M_V$ is not an LB subgraph, though
it contains a $G_{IV}(6)$ induced by $\{x,y,z,0,1,3\}$.
Similar issues arise in the completion of Tucker
submatrices that are examples of $M_{II}(k)$ and $M_{III}(k)$.
}\label{fig:M4M5}
\end{figure}

If 0 and 1 are adjacent but 1 and 2 are not, then completing
the $M_{IV}$ yields the clique matrix of $G_I$.  However,
the subgraph induced by rows in $G$
is second graph on the right.  This is not an LB subgraph.  
When $z$ is deleted from it, however,
it yields an LB subgraph, $G_{IV}(6)$.  The case where
0 and 1 are nonadjacent and 1 and 2 are adjacent is symmetric with this
case.

If 0 and 1 are adjacent and 1 and 2 are adjacent,
the completion still gives the clique matrix of $G_I$,
but the subgraph of $G$ induced by 
its rows is the third graph on
the right, which is an instance of $G_{IV}(7)$.

Similarly, if an instance of $M_V$ is found, its
completion yields a submatrix of the clique matrix of $G$
that is the clique matrix of $G_{II}$.
However, the subgraph of $G$ induced by 
the rows of the instance could be
the second graph on the right, which is not an LB graph,
though it contains an induced LB subgraph, an instance
of $G_{IV}$.

Similar issues arise with the completion of instances of
$M_I(k)$, $M_{II}(k)$ and $M_{III}(k)$.    

To get around these problems, we make use of a stronger version of 
Theorem~\ref{thm:TuckerTime}:

\begin{lemma}\label{lem:TuckMinimal}
When $M$ does not have the consecutive-ones property,
it takes $O(size(M))$ time to find an instance
of a Tucker submatrix whose rows are a minimal
set of rows of $M$ that contain a Tucker submatrix of $M$.
\end{lemma}
\begin{proof}
If {\tt TuckerSubmatrix} (Algorithm~\ref{alg:TuckerSubmatrix})
returns a submatrix with at most four rows, this follows
from Lemma~\ref{lem:TuckerRowsCorrect}.
Otherwise, it follows from Lemma~\ref{lem:FindRowsCorrect}.
\end{proof}

The key result of this section is the following:

\begin{lemma}\label{lem:completion}
Suppose $G$ is chordal,
$M_T$ is an instance of a Tucker submatrix
in the clique matrix $M$ of $G$, and that the rows
of $M_T$ are a minimal set of rows that contain a Tucker
submatrix of $M$.  Then $M_T$ is not an instance of $M_I(k)$ for $k \geq 4$,
hence it has three incomplete columns.
For every choice of rows $x,y,z$ that
complete the three incomplete columns of $M_T$,
$x,y,z$ and the
rows of $M_T$ induce an LB subgraph of $G$.
\end{lemma}

We prove the lemma below.  This gives the strategy
for finding an LB subgraph in
any graph that is not an interval graph, which is summarized
as Algorithm~\ref{alg:FindLBSubgraph}.

\begin{algorithm}
\KwData{$G$ is not an interval graph.}
\KwResult{The vertices of $G$ inducing an LB subgraph.}
\begin{enumerate}
\item Test whether $G$ is chordal using the 
algorithm of~\cite{RTL:triangulated}.

\item If it is not chordal, return an instance
of $G_{III}(k)$ for $k \geq 4$, using the algorithm of~\cite{TarjYan85}.

\item Otherwise, let $M$ be the clique
matrix of $G$ produced by the algorithm 
of~\cite{RTL:triangulated}.

\item Using a call to {\tt TuckerSubmatrix$(M)$}, find
an instance $M_T$ of a Tucker submatrix in $M$.
Let $X$ be its rows in $M$.

\item Find three rows $\{x,y,z\}$ of $M$ that complete $M_T$.

\item Return $X \cup \{x,y,z\}$.
\end{enumerate}

\caption{\texttt{FindLBSubgraph($G$)}\label{alg:FindLBSubgraph}}
\end{algorithm}

It remains 
to prove Lemma~\ref{lem:completion}.  We begin by showing that
the three incomplete columns in
any instance of a Tucker submatrix in the clique matrix of
a chordal graph can always be completed with three 
additional rows.

A {\em clique tree} of a chordal graph is a tree $\mT$ that has one node
for each maximal clique, and with the property that, for each
vertex $v$ of $G$, the cliques that contain $v$ induce a connected
subtree.
Every connected chordal graph has a clique tree, see for example~\cite{Gol80}.
There is not necessarily a unique clique tree.

Observe that if $\mK$ is the set of cliques that contain some vertex $v$
and a clique $C$ does not contain $v$, $C$ cannot lie on the path
between any pair of members of $\mK$ in any clique tree.  Otherwise,
the subtree induced by cliques containing it would not be connected,
contradicting the definition of a clique tree.

Generalizing from this insight, we obtain the following:

\begin{lemma}\label{lem:Cleaf}
Let $G$ be a connected chordal graph and let $\mT$ be a clique tree for $G$.
Let $\mK$ be a set of cliques of $G$
and let $C$ be a clique such that $C \not\in \mK$.  
Let $\mK'$ be the multiset obtained
by removing the members of $C$ from every clique in $\mK$.
If the intersection graph of $\mK'$ is connected, then $C$
does lie on the path in $\mT$ between any two members of $\mK$.
\end{lemma}
\begin{proof}
Suppose to the contrary that $C$ lies on the path between
two members of $\mK$ in $\mT$.
Removal
of $C$ from $\mT$ leaves a set of two or more trees
that partition $\mK$.  Since the intersection graph of $\mK'$
is connected, there exist two of these trees, one with
$K_1 \in \mK'$ and the other with $K_2 \in \mK'$ such that
$K_1$ intersects $K_2$ on a vertex $v \not\in C$.
Since $C$ lies on the path from $K_1$ to $K_2$ in $\mT$,
the subtree of the clique
tree induced by cliques containing $v$ is not connected,
contradicting the definition of a clique tree.
\end{proof}

\begin{figure}[!ht]
  \captionsetup[subfigure]{labelformat=empty}
  \centering
\subfloat[A chordal graph $G$]{\includegraphics[scale=.3]{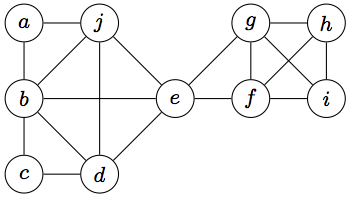}}\\

\subfloat[$G$'s clique tree]{\includegraphics[scale=.3]{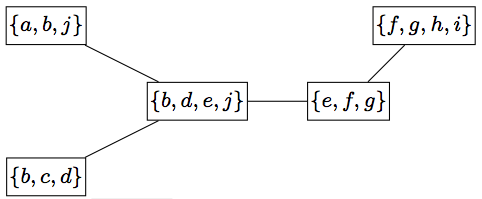}} \quad \quad \quad
\subfloat[Intersection graph of $\mathcal{K} \setminus C$]{\includegraphics[scale=.3]{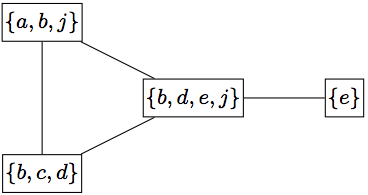}}
\caption{Let $C$
be the clique $\{f,g,h,i\}$ and let $\mK$ be the remaining cliques of $G$.
To the lower right
is the intersection graph of the members of $\mK$ after elements of $C$
have been removed from them.  This graph is connected,
which means by Lemma~\ref{lem:Cleaf} that $\{f,g,h,i\}$ cannot lie on the path
between any pair of members of $\mK$ in any clique tree of $G$.  
}\label{fig:Cleaf}
\end{figure}

Figure~\ref{fig:Cleaf} illustrates the idea.  
The following lemma
is immediate from results that appear in~\cite{Gol80}.

\begin{lemma}\label{lem:leafDeletion}
Let $\mT$ be a clique tree for a chordal graph $G$ and let $K$ be a leaf.  Then $K$ contains
a simplicial vertex of $G$.  Let $S$ be the simplicial vertices of $K$ and
let $\mT'$ be the result of deleting leaf $K$ from $\mT$. 
Deleting $S$ from $G$ yields an induced subgraph that has
$\mT'$ as a clique tree.  (See Figure~\ref{fig:shrink}.)
\end{lemma}

\begin{figure}[!ht]
  \captionsetup[subfigure]{labelformat=empty}
  \centering
\subfloat[A chordal graph $G$]{\includegraphics[scale=.3]{cleaf1.png}} \quad \quad \quad
\subfloat[$G$'s clique tree $\mathcal{T}$]{\includegraphics[scale=.3]{cleaf2.png}}\\

\subfloat[$G - \{h,i\}$]{\includegraphics[scale=.3]{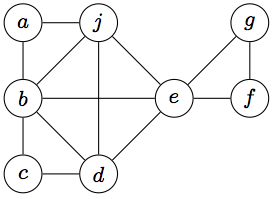}}
\quad \quad \quad
\subfloat[$G - \{h,i\}$'s clique tree]{\includegraphics[scale=.3]{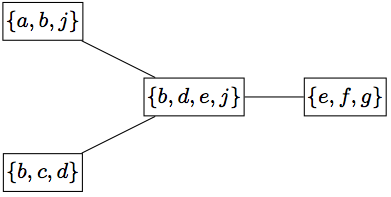}}\\

\caption{``Shrinking'' a clique tree.  If $G$ is a chordal graph and $\mT$
is a clique tree for it, then deleting the simplicial vertices $\{h,i\}$ in a leaf $\{f,g,h,i\}$
of $\mT$ gives a smaller graph that has as a clique tree the tree obtained by deleting
$K$ from $\mT$.}\label{fig:shrink}
\end{figure}

\begin{definition}\label{def:shrinking}
By {\em shrinking a clique tree $\mT$}, let us denote
the operation of deleting the set $S$ of simplicial vertices in a leaf
$K$ of $\mT$, yielding a smaller graph $G'$ that has $\mT - K$
as a clique tree.
\end{definition}

\begin{lemma}\label{lem:TuckerSimplicials}
Let $G$ be a chordal graph and let $M_T$ be a Tucker submatrix of a clique matrix $M$ of $G$.
Then $M_T$ can be completed using rows of $M$.
\end{lemma}

\begin{proof}
Let $C$ be the column of $M$ that contains an incomplete column 
of an instance of a Tucker matrix $M_T$, and let
$\mK$ be the columns of $M$ that
contain the the remaining columns of the instance of $M_T$.
Iteratively shrink the clique tree of $G$ subject
to the constraint that we do not delete $C$ or any member of $\mK$.
Let $\mT'$ be the resulting
clique tree when no more cliques can be deleted without violating
the constraint, and let $G'$ be the corresponding induced
subgraph of $G$.  

All leaves of $\mT'$ are members of $\mK$;
otherwise they could be deleted through further
shrinking without violating
the constraint.  Therefore, any member of $\mK$ that
is not a leaf lies on the path in $\mT'$ between two
other members of $\mK$.  

By inspection on each of $M_I$, $M_{II}$, $M_{III}(3)$, $M_{IV}$
and $M_V$, each choice of $C$ satisfies the requirements
of Lemma~\ref{lem:Cleaf} that prevent it from lying on the path between any two
members of $\mK$.  Therefore, each incomplete column of $M_T$
is a leaf in $\mT'$, hence contains a simplicial vertex $v$ in $G'$.
The row corresponding to $v$ has a single 1 in the clique matrix
of $G'$.  Therefore, the addition of $v$ to $M_T$ completes the column.
\end{proof}

\begin{lemma}\label{lem:simplicialIS}
If three vertices complete three columns of a Tucker submatrix
of the clique matrix of a chordal graph $G$, then they are
pairwise nonadjacent.
\end{lemma}
\begin{proof}
Let $X$ be the vertices of the Tucker submatrix.
For each pair $a,b$ of the vertices that complete it, the adjacencies
of $a$ and $b$ to members of $X$ are known, by definition.  Also, $a$
has a neighbor $a' \in X \setminus N(b)$ and $b$ has a neighbor
$b' \in X \setminus N(a)$, and there is a path $P$ in $X$ from $a'$ to $b'$
that avoids the common neighbors of $a$ and $b$.  Even if there are 
unknown chords on $P$, $a$, $a'$, $b'$, $b$ and zero or more members
of $P$ form a chordless cycle, contradicting the chordality of $G$.
\end{proof}

{\bf Proof of Lemma~\ref{lem:completion}}.
The cases for $M_{IV}$ and $M_V$ are illustrated in Figure~\ref{fig:M4M5}.
By assumption in each case, $\{0,1,2,3\}$ is a minimal set of rows of $M$
that contain an instance of a Tucker matrix.
Also, assume without loss of generality that
the depicted the ordering of columns is consistent with a consecutive-ones
ordering of rows $\{0,1,2\}$ of $M$.  
By Lemma~\ref{lem:simplicialIS}, the only possible adjacencies
among $\{0,1,2,3\}$ that are not reflected by the intersection graph 
of $M_T$ are $\{0,1\}$ and $\{1,2\}$.
If only 0 and 1 are adjacent, 
$\{0,1,3,x,y,z\}$ induces an instance of $G_{IV}(6)$.  Since the clique matrix
of $G_{IV}(6)$ consists of three rows for its simplicial
vertices, plus a Tucker matrix, rows
$\{0,1,3\}$ induce a Tucker matrix in the clique matrix of $G_{IV}(6)$.
Thus, it is a Tucker matrix in the clique matrix of $G$,
contradicting the minimality of $\{0,1,2,3\}$.  This case cannot
occur.  By symmetry, it cannot be the case that only $1$ and $2$ are
adjacent.  Therefore, $x$, $y$ and $z$,
together with the rows of the $M_V$, induce
an instance of $G_I$ or of $G_{IV}$.

Similarly, for $M_V$, by Lemma~\ref{lem:simplicialIS},
the case where 0 and 1 are adjacent
yields an $M_{II}(3)$ on $\{0,1,3\}$, contradicting the minimality
of $\{0,1,2,3\}$, hence, the completion with $x$, $y$ and $z$
induces a $G_{II}$.

For $M_T \in \{M_I(k), M_{II}(k), M_{III}(k)\}$, the rows
of $M$ corresponding to $\{0,$ $1,$ $ 2, \ldots, k-1\}$
are a minimal set $\mR$ of rows of $M$ that contain an instance of a Tucker
submatrix, by assumption.  Let $X$ be the corresponding vertices of $G$.
Assume that the rows of $\mR$ are ordered as shown in Figure~\ref{fig:TuckMatrices},
so that only the last row of $\mR$ is not consecutive-ones ordered.

Every pair of nonadjacent vertices in the intersection graph
of $M_T$ contains some
$i \in \{1,2, \ldots, k-4\}$ if $M_T$ is an instance of $M_{II}(k+1)$,
or $i \in \{0,1, \ldots, k-4\}$ if $M_T$ is an instance of $M_{III}(k)$
for $k \geq 4$.  Without loss of generality, we may assume that
$i$ is the lower-numbered member of the adjacent pair.
In each case, row $i+1$ contains the right endpoint of row $i$
in $M$.  Therefore,
$i+2$ the only possible higher-numbered neighbor of $i$ in $G$ that is not a neighbor
in the intersection graph of $M_T$.  Suppose $i+2$
is a neighbor of $i$.  Then $M$ has a column that contains both $i$ and $i+2$.
In $M_T$, we may replace the column that contains $i$ and $i+1$
and the column that contains $i+1$ and $i+2$ with this column,
and then delete row $i+1$ from $M_T$,
yielding a smaller instance of $M_{II}$ or of $M_{III}$.  This contradicts
the assumed minimality of the rows of $M_T$.  It follows that 
the intersection graph of rows of $M_T$ faithfully represents $G[X]$.
The completion of $M_T$ with three vertices $\{x,y,z\}$ induces an
induced LB subgraph by Lemma~\ref{lem:simplicialIS}.

An identical argument applies to $i \in \{0,1, \ldots, k-4\}$ for $M_I(k)$
and $k \geq 4$.  By the cyclic symmetry of $M_I$, it therefore applies for all rows
of $M_I$, and the intersection graph of $M_I$ faithfully represents $G[X]$.
Since this is a chordless cycle, no instance of $M_I(k)$ for $k \geq 4$
satisfies the conditions of Lemma~\ref{lem:TuckMinimal}.
The intersection graph of $M_I(3)$ is complete, so it faithfully
represents $G[X]$, and its completion is an induced LB subgraph
by Lemma~\ref{lem:simplicialIS}.

Summarizing the results of this section, we obtain the following:

\begin{theorem}\label{thm:LBTime}
Given an adjacency-list representation of an arbitrary graph $G$, 
it takes $O(n+m)$ time to find either an interval model of $G$
or an induced LB subgraph.
\end{theorem}
\begin{proof}
If $G$ is an interval graph, an interval model can be produced
in linear time by~\cite{BL76}.
Otherwise, in a call to {\tt FindLBSubgraph$(G)$},
the cited algorithms run in $O(n+m)$ time.  
When $G$ is not chordal, it returns an instance of $G_{III}(k)$
for $k \geq 4$, which is an LB subgraph.
When $G$ is chordal, the clique matrix $M$ returned by the
algorithm of~\cite{RTL:triangulated} has $size(M) = O(n+m)$.
Since $G$ is not an interval graph, $M$ does not have the consecutive-ones
property, so it contains a Tucker submatrix.
By Lemma~\ref{lem:TuckMinimal}, it takes
in $O(size(M)) = O(n+m)$ time to find one that satisfies
the conditions of Lemma~\ref{lem:completion}.  Denote it by $M_T$.

Let $X$ be the vertices corresponding to rows
of $M_T$.  For each incomplete column of $M_T$, make
a list of the rows that have a 1 in the column.
For each vertex of $G$ other than those in $X$,
check whether $N(v) \cap X$ is the set of rows in one of these
lists.  This takes $O(1+|N(v)|)$ by 
marking and counting of neighbors of $v$ in $X$, checking each list for an unmarked
vertex, then unmarking them, for a total of $O(n+m)$ over all
vertices of $G$.  
This gives the set of rows that complete incomplete
columns of $M_T$, since a row that has a 1 in any
other column of $M_T$ has neighbors in $X$ that
are not in any incomplete column.  Selecting one row from
each of these sets gives an LB subgraph by Lemma~\ref{lem:completion}.
\end{proof}

\bibliographystyle{plain}
\bibliography{to}
\end{document}